\documentclass[aip,jmp,amsmath,amssymb,reprint]{revtex4-1}

\usepackage{amsmath,amsthm,amsfonts,amssymb,amscd} 
\usepackage{indentfirst}
\usepackage{graphicx}
\usepackage{color}
\usepackage[american]{babel}
\usepackage{float}
\usepackage[dvipsnames]{xcolor}
\usepackage[colorlinks=true,citecolor=ForestGreen,linkcolor=magenta,urlcolor=Mulberry,hyperindex]{hyperref}
\usepackage{cleveref}
\usepackage{mathtools}
\usepackage{bm}
\usepackage{bbold}
\usepackage{units}
\usepackage{appendix}
\usepackage{comment}
\usepackage{etoolbox}
\usepackage{soul}

\makeatletter
\def\@email#1#2{%
 \endgroup
 \patchcmd{\titleblock@produce}
  {\frontmatter@RRAPformat}
  {\frontmatter@RRAPformat{\produce@RRAP{*#1\href{mailto:#2}{#2}}}\frontmatter@RRAPformat}
  {}{}
}%
\makeatother

\newcommand{\ket}[1]{\vert#1\rangle}
\newcommand{\bra}[1]{\langle#1\vert}

\newcommand{\ketbra}[2]{\vert #1 \rangle \langle #2 \vert}
\newcommand{\avg}[1]{\langle#1\rangle}
\newcommand{\id}{\mathbb{1}}

\newcommand{\Lcal}{\mathcal{L}}
\newcommand{\Hcal}{\mathcal{H}}

\newcommand{\Ecal}{\mathcal{E}}

\newcommand{\Scal}{\mathcal{S}}
\newcommand{\Wcal}{\mathcal{W}}

\newcommand{\tr}{\text{\normalfont Tr}}

\newtheorem{theorem}{Theorem}%[chapter]
\newtheorem*{theorem*}{Theorem}
\newtheorem{definition}{Definition}%[chapter]
\newtheorem{lemma}{Lemma}%[chapter]
\newtheorem{example}{Example}%[chapter]
\newcommand{\iqoqi}{Institute for Quantum Optics and Quantum Information (IQOQI), Austrian Academy of Sciences, Boltzmanngasse 3, A-1090 Vienna, Austria}
\newcommand{\univie}{Vienna Center for Quantum Science and Technology (VCQ), Faculty of Physics, University of Vienna, Boltzmanngasse 5, A-1090 Vienna, Austria}
\newcommand{\todai}{Department of Physics, Graduate School of Science, The University of Tokyo, Hongo 7-3-1, Bunkyo-ku, Tokyo 113-0033, Japan}

\newcommand{\map}[1]{\widetilde{#1}} 
\renewcommand{\H}{\mathcal{H}}
\renewcommand{\L}{\mathcal{L}}
\newcommand{\dbra}[1]{\langle\hspace{-.8mm}\langle #1\vert}
\newcommand{\dket}[1]{\vert#1\rangle\hspace{-.8mm}\rangle}
\newcommand{\dketbra}[2]{\vert #1 \rangle \hspace{-.8mm} \rangle \hspace{-.4mm} \langle\hspace{-.8mm}\langle #2 \vert}
\begin{document}

\title{Unitary channel discrimination beyond group structures: \\ Advantages of sequential and indefinite-causal-order strategies}

\author{Jessica Bavaresco}
\email{jessica.bavaresco@oeaw.ac.at}
\affiliation{\iqoqi}
%\affiliation{\tu}

\author{Mio Murao}
\affiliation{\todai}
\affiliation{Trans-scale Quantum Science Institute, The University of Tokyo, Hongo 7-3-1, Bunkyo-ku, Tokyo 113-0033, Japan}

\author{Marco T\'ulio Quintino}
\affiliation{\iqoqi}
\affiliation{\univie}

\date{27th May 2021}

\begin{abstract}
For minimum-error channel discrimination tasks that involve only unitary channels, we show that sequential strategies may outperform the parallel ones. Additionally, we show that general strategies that involve indefinite causal order are also advantageous for this task. However, for the task of discriminating a uniformly distributed set of unitary channels that forms a group, we show that parallel strategies are, indeed, optimal, even when compared to general strategies. We also show that strategies based on the quantum switch cannot outperform sequential strategies in the discrimination of unitary channels. Finally, we derive an absolute upper bound for the maximal probability of successfully discriminating any set of unitary channels with any number of copies for the most general strategies that are suitable for channel discrimination. Our bound is tight since it is saturated by sets of unitary channels forming a group $k$-design.
\end{abstract}

\maketitle

%%%%%%%%%%%%%%%%%%%%%%%%%%%%%%%%%%%%%%%%%%%%%%%%%%%%%%%%%%%%%%%%%
%%%%%%%%%%%%%%%%%%%%%%%%%%%%%%%%%%%%%%%%%%%%%%%%%%%%%%%%%%%%%%%%%

\section{Introduction}

The discrimination of different hypotheses is a fundamental part of the scientific method that finds application in the most distinct areas, such as information theory,~\cite{information_theory} bioinformatics,~\cite{bioinformatics} machine learning,~\cite{pearl2014probabilistic} and behavioral and social sciences.~\cite{psychological_methods} In a discrimination task, one seeks for the best manner to decide whether a particular hypothesis is the most likely to be the best description of some scenario or experiment. An important, albeit general, instance of a discrimination task consists in identifying between different input-output relations or different causal-effect dynamics a physical system may undergo. For instance, in an x-ray examination that identifies whether a person has a broken bone, the examiner prepares an initial state, which is subjected to certain dynamics when passing through the person's body, before being measured by some physical apparatus. In this case, the goal is to distinguish between the dynamics related to a broken and a healthy bone by implementing the most appropriate input state and measurement apparatus.

In its most fundamental level, closed-system dynamics in quantum theory are described by unitary operations. Hence, being able to discriminate between different unitary operations is a ubiquitous task within quantum theory and quantum technologies. Examples of tasks directly related to our ability to discriminate unitary operations are unitary equivalence determination,~\cite{shimbo18,soeda21} quantum metrology,~\cite{giovannetti06,metrology} quantum hypothesis testing,~\cite{hayashi06} quantum parameter estimation,~\cite{paris09} alignment and transmission of reference frames,~\cite{chiribella04b,bartlett07} and discrimination and tomography of quantum circuit elements.~\cite{chiribella07} 

Discrimination tasks are also relevant to the field of computer science. An oracle, which is an abstract machine used to study decision problems, may be understood as a black box that solves certain problems with a single operation. From a quantum computational perspective, a quantum oracle is a unitary operation whose internal mechanisms are unknown and are employed in seminal quantum algorithms, such as the Deustch-Josza algorithm,~\cite{deutsch92} Grover's algorithm,~\cite{grover98} and Simon's algorithm.~\cite{simon94} These oracle-based quantum algorithms may be recast as unitary discrimination tasks.~\cite{chefles07}

Such practical and fundamental interest has motivated an extensive study of the discrimination of unitary channels within the context of quantum information theory, leading to a plethora of interesting results. 

Contrarily to the problem of quantum state discrimination,~\cite{helstrom69} in which two states cannot be perfectly distinguished with a finite number of uses, or copies, unless they are orthogonal, it has been remarkably shown that any pair of unitary channels can indeed always be perfectly distinguished with a finite number of copies.~\cite{acin01,dariano01} Moreover, perfect discrimination of a pair of unitary channels can always be achieved by a parallel scheme~\cite{acin01,dariano01} (see also Ref.~\cite{duan07}). Even when perfect discrimination is not possible, sequential strategies can never outperform parallel strategies in a task of discrimination between a pair of unitary channels.~\cite{chiribella08-1} Concerning the discrimination of sets of more than two unitary channels, when considering unitaries which are a representation of a group and uniformly distributed, Ref.~\cite{chiribella08-1} showed once more that, for any number of copies, sequential strategies are not advantageous when compared to parallel strategies. For related tasks such as error-free and unambiguous unitary channel discrimination,~\cite{chiribella13-02} unitary estimation,~\cite{chiribella08-1} unitary learning,~\cite{bisio10} and unitary store-and-retrieve,~\cite{sedlak19} parallel strategies were also proven to be optimal. Up to this point, no unitary channel minimum-error discrimination tasks in which sequential strategies outperform parallel strategies are known, to the extent of our knowledge.

In this work, we focus on the discrimination of sets of more than two unitary channels with multiple copies to study the potential advantages that different classes of strategies can bring to this task. The first contribution of our work is precisely to show examples of discrimination tasks of unitary channels in which sequential strategies are advantageous, when compared to parallel strategies that allow the same number of copies. In fact, contrarily to the tasks of error-free and unambiguous unitary channel discrimination,~\cite{chiribella13-02} we show that sequential strategies can achieve perfect discrimination in tasks that parallel strategies cannot. 

Then, motivated by the recent advances in channel discrimination theory that have established the advantage of general discrimination strategies that involve indefinite causal order for general channels,~\cite{bavaresco20} we study the potential advantages of these general strategies for the specific case of unitary channel discrimination. 

Extending the framework developed in Ref.~\cite{bavaresco20} to discrimination tasks that allow for the use of multiple copies, we achieve the following results. We prove the optimality of parallel strategies, even when compared against general strategies, in tasks of discrimination of uniformly distributed unitary channels that form a unitary representation of some group for any number of copies. However, the power of general strategies is revealed when applied to discrimination tasks that fail to satisfy at least one of these requirements. In these cases, we show that general indefinite-causal-order strategies can outperform sequential strategies, and again that sequential strategies can outperform the parallel ones.  Then, we show that a particular case of general strategies, one that applies processes related to the quantum switch~\cite{chiribella13-01} and its generalizations,~\cite{araujo14,yokojima20,barrett20} can never outperform sequential strategies in the discrimination of unitary channels. 

The final contribution of our work is to derive an ultimate upper bound for the maximal probability of successful discrimination of any ensembles of unitary channels with a uniform probability distribution. Our result represents an upper bound for the most general strategy that can possibly be employed in a task of channel discrimination. We show that this bound is saturated by parallel strategies for the discrimination of unitary groups that form a $k$-design, where $k$ is the number of allowed copies.

%%%%%%%%%%%%%%%%%%%%%%%%%%%%%%%%%%%%%%%%%%%%%%%%%%%%%%%%%%%%%%%%%
%%%%%%%%%%%%%%%%%%%%%%%%%%%%%%%%%%%%%%%%%%%%%%%%%%%%%%%%%%%%%%%%%

\section{Minimum-error channel discrimination} 

In a task of minimum-error channel discrimination, one is given access to an unknown quantum channel $\map{C}_i:\Lcal(\Hcal^I)\to\Lcal(\Hcal^O)$, 
which maps quantum states from an input linear space $\Hcal^I$ to an output linear space $\Hcal^O$. This quantum channel is known to have been drawn with probability $p_i$ from a known ensemble of channels $\Ecal=\{p_j,\map{C}_j\}_{j=1}^N$. The task is to determine which channel from the ensemble was received using a limited amount of uses/queries of it, which is essentially to determine the classical label $i$ of channel $\map{C}_i$.
In order to accomplish this task in the case where only a single use of the unknown channel is allowed, one may send part of a potentially entangled state $\rho\in\Lcal(\Hcal^I\otimes\Hcal^\text{aux})$ through the channel $\map{C}_i$ and subsequently jointly measure the output state with a positive-operator valued measure (POVM) $M=\{M_a\}_{a=1}^N,M_a\in\Lcal(\Hcal^O\otimes\Hcal^\text{aux})$. When both the state and measurement are optimized according to the knowledge of the ensemble, the outcome of the measurement will correspond to the most likely value of the label $i$ of the unknown channel. Then, the maximal probability of successfully determining which channel is at hand is given by 
\begin{equation}\label{eq::P_rhoM}
P \coloneqq \max_{\rho,\{M_i\}}\sum_{i=1}^N p_i \tr\left[(\widetilde{C}_i\otimes\widetilde{\id})(\rho)\,M_i\right], 
\end{equation}
where $\widetilde{\id}:\Lcal(\Hcal^\text{aux})\to\Lcal(\Hcal^\text{aux})$ is the identity map.

When more than one use, or copy, as we will refer to from now on, is allowed, different strategies come into play, each exploring a different order in which the copies of the unknown channel are applied. In Figs.~\ref{fig::realization}(a) and (b), we illustrate two such possibilities, a parallel and a sequential strategy, respectively.
However, a more general strategy can be defined by considering the most general higher-order transformation that can map $k$ quantum channels to a valid probability distribution [see Fig.~\ref{fig::realization}(c)]. It has been shown that some of these general strategies may employ processes with an indefinite causal order and that these strategies may outperform parallel and sequential ones in tasks of channel discrimination.~\cite{bavaresco20}

%%%%%%%%%%%%%%%%%%%%%%%%%%%%%%%%%%%%%%%%%%%%%%%%%%%%%%%%%%%%%%%%%
%%%%%%%%%%%%%%%%%%%%%%%%%%%%%%%%%%%%%%%%%%%%%%%%%%%%%%%%%%%%%%%%%

\section{Tester formalism} 

To facilitate the approach to this problem, a concise and unified formalism of testers,~\cite{chiribella09} also referred to as process POVMs,~\cite{ziman08,ziman10} was developed in Ref.~\cite{bavaresco20}, providing practical tools for both the comparison between different strategies and for the efficient computation of the maximal probability of successful discrimination of a channel ensemble under different classes of strategies. We now revise the tester formalism while extending its definitions to strategies that involve a finite number of copies $k$ of the unknown channel.

In order to apply the tester formalism, we will make use of the Choi-Jamio\l{}kowski (CJ) representation of quantum maps. The CJ isomorphism is a one-to-one correspondence between completely positive maps and positive semidefinite operators, which allows one to represent any linear map $\widetilde{L}:\Lcal(\Hcal^I)\mapsto\Lcal(\Hcal^O)$ by a linear operator $L\in\Lcal(\Hcal^I\otimes\Hcal^O)$ defined by 
\begin{equation}
L \coloneqq (\widetilde{\id}\otimes\widetilde{L})(\Phi^+), 
\end{equation}
where $\Phi^+=\sum_{ij}\ket{ii}\bra{jj} \in \Lcal(\Hcal^I\otimes\Hcal^I)$, with $\{\ket{i}\}$ being an orthonormal basis, is an unnormalized maximally entangled state. In this representation, a quantum channel, i.e., a completely positive trace-preserving (CPTP) map $\widetilde{C}:\Lcal(\Hcal^I)\mapsto\Lcal(\Hcal^O)$, is represented by a linear operator $C\in\Lcal(\Hcal^I\otimes\Hcal^O)$, often called the ``Choi operator'' of channel $\widetilde{C}$, which satisfies 
\begin{align}
    C\geq0, \ \ \ \tr_O C = \id^I,
\end{align}
where $\tr_O$ denotes the partial trace over $\Hcal^O$ and $\id^I$ denotes the identity operator on $\Hcal^I$. In particular, the Choi operator of a unitary channel is proportional to a maximally entangled state. Using Choi operators of quantum channels, we can equivalently represent the channel ensemble $\{p_i,\widetilde{C}_i\}_{i=1}^N$ as $\{p_i,C_i\}_{i=1}^N$, where $C_i$ is the Choi operator of channel $\widetilde{C}_i$.

A tester is a set of positive semidefinite operators $T=\{T_i\}_{i=1}^N, T_i\in\Lcal(\Hcal^I\otimes\Hcal^O)$, which obey certain normalization constraints and which, when {taking} the trace with the Choi operator of a quantum channel $C$, lead to a valid probability distribution, according to $p(i|C)=\tr\left(T_i\,C\right)$. In this sense, testers act on quantum channels similarly to how POVMs act on quantum states and can, therefore, be interpreted as a ``measurement'' of a quantum channel. Testers allow us to rewrite the maximal probability of successful discrimination of the channel ensemble $\Ecal=\{p_i,C_i\}_{i=1}^N$ as 
\begin{equation}
P = \max_{\{T_i\}}\sum_{i=1}^N p_i \tr\left(T_i\,C_i\right).
\end{equation}

The advantage of this representation is the simplification of the optimization problem that defines the maximal probability of success: now, optimization over different discrimination strategies may be achieved by maximizing $P$ over the set of valid testers, as opposed to optimizing over both states and measurements [as in Eq.~\eqref{eq::P_rhoM}]. 

\begin{figure*}%[h!]
\begin{center}
	\includegraphics[width=0.9\textwidth]{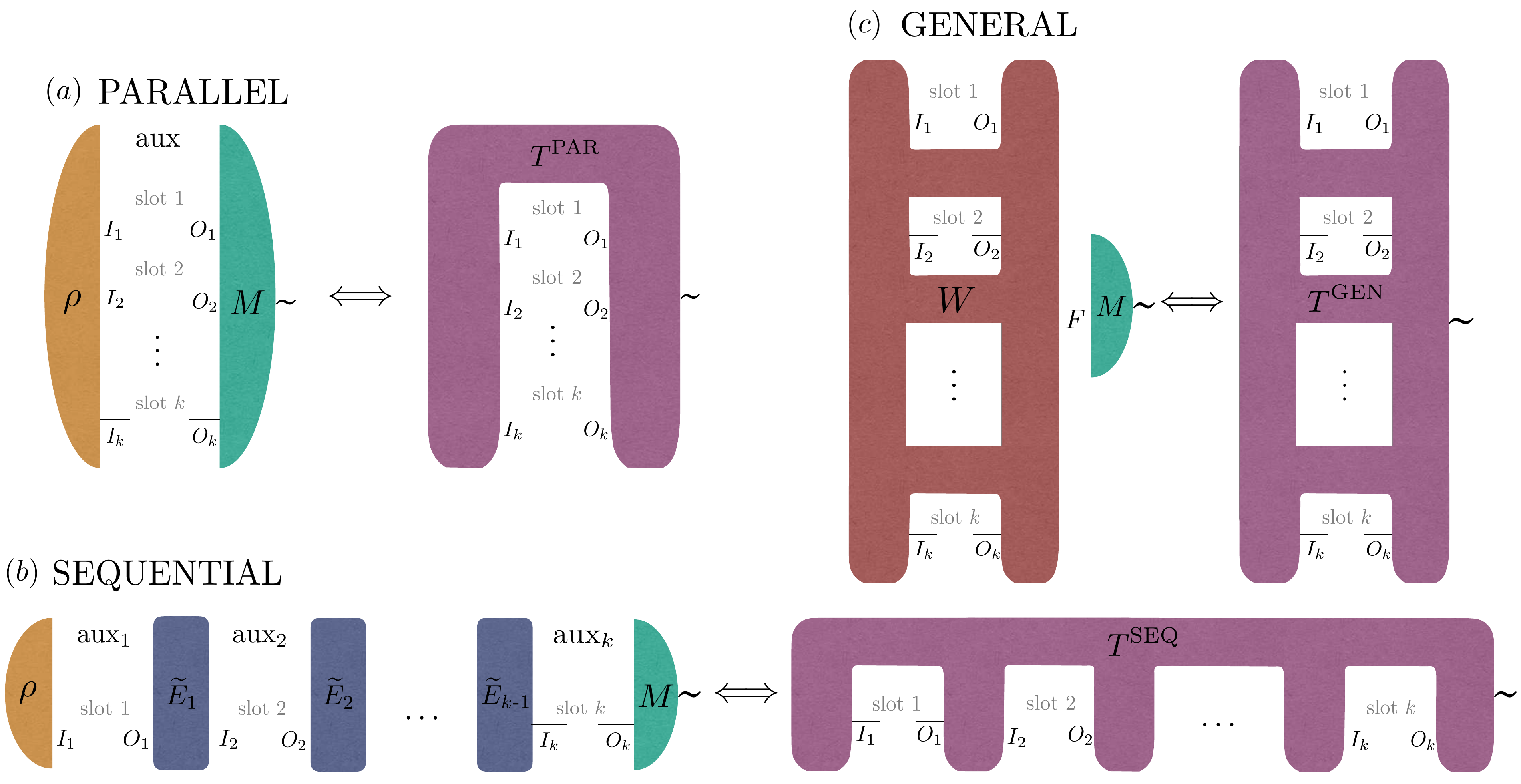}
	\caption{Schematic representation of the realization of every $k$-copy (a) parallel tester $T^\text{PAR}$ with a state $\rho$ and a POVM $M$; (b) sequential tester $T^\text{SEQ}$ with a state $\rho$, channels $\widetilde{E_i}$, $i\in\{1,k-1\}$, and a POVM $M$; and (c) general tester $T^\text{GEN}$ with a process matrix $W$ and a POVM $M$.}
\label{fig::realization}
\end{center}
\end{figure*}

For the case of $k$ copies, different normalization constraints define testers that represent different classes of strategies. We now define $k$-copy testers that represent parallel, sequential, and general strategies.

Let $\Hcal^{I}\coloneqq\bigotimes_{i=1}^k \Hcal^{I_i}$ and $\Hcal^{O}\coloneqq\bigotimes_{i=1}^k \Hcal^{O_i}$ be the joint input and output spaces, respectively, of $k$ copies of a quantum channel, let $d_I\coloneqq\text{dim}(\Hcal^{I})$ and $d_O\coloneqq\text{dim}(\Hcal^{O})$ be their dimension, and let 
\begin{equation}
_{X}(\cdot) \coloneqq \tr_{X}(\cdot)\otimes\frac{\id^{X}}{d_X}
\end{equation}
denote a trace-and-replace operation in $\Hcal^{X}$.

\subsection{Parallel strategies}

Parallel strategies are the ones that consist of sending each system that composes a multipartite state through one of the copies of the unknown channel, in such a way that the output of each copy does not interact with the input of the others, and jointly measuring the output state at the end [see Fig.~\ref{fig::realization}(a), left].
These strategies are characterized by parallel testers [see Fig.~\ref{fig::realization}(a), right] $T^\text{PAR}=\{T^\text{PAR}_i\}_{i=1}^N$, $T^\text{PAR}_i\in\Lcal(\Hcal^{I}\otimes\Hcal^{O})$. Let $W^\text{PAR}\coloneqq \sum_i T^\text{PAR}_i$. Then, parallel testers are defined as 
\begin{align}
    T^\text{PAR}_i &\geq 0 \ \ \ \forall \, i \\
    \tr \, W^\text{PAR} &= d_O \\
    W^\text{PAR} &= _{O}W^\text{PAR}.
\end{align}
This is equivalent to defining $W^\text{PAR}$ to be a parallel process, satisfying $W^\text{PAR}=\sigma^{I}\otimes\id^O$, where $\tr(\sigma^{I})=1$.

\subsection{Sequential strategies}

In a sequential strategy, a quantum system is sent through the first copy of the channel, and its output system is allowed to be sent as input of the next copy, while general CPTP maps may act on the systems in between copies. The final output is measured by a POVM [see Fig.~\ref{fig::realization}(b), left]. These strategies are represented by sequential testers [see Fig.~\ref{fig::realization}(b), right] $T^\text{SEQ}=\{T^\text{SEQ}_i\}_{i=1}^N, T^\text{SEQ}_i\in\Lcal(\Hcal^{I}\otimes\Hcal^{O})$. Let $W^\text{SEQ}\coloneqq \sum_i T^\text{SEQ}_i$. Then, sequential testers are defined as 
\begin{align}
    T^\text{SEQ}_i &\geq 0 \ \ \ \forall \, i \\
    \tr \, W^\text{SEQ} &= d_O \\
    W^\text{SEQ} &= _{O_k}W^\text{SEQ}\\
    _{I_k O_k}W^\text{SEQ} &= _{O_{(k-1)} I_k O_k}W^\text{SEQ}\\
    &\ldots \nonumber \\
    _{I_2 O_2 \ldots I_k O_k}W^\text{SEQ} &= _{O_1 I_2 O_2 \ldots I_k O_k}W^\text{SEQ}.
\end{align}
This is equivalent to defining $W^\text{SEQ}$ to be a $k$-slot comb,~\cite{chiribella09} or a $k$-partite ordered process matrix~\cite{araujo15} (see also Refs.~\cite{gustoki06,kretschmann05}).

\subsection{General strategies}

Finally, general strategies are defined by the most general higher-order operations that can transform $k$ quantum channels into a joint probability distribution. They can be regarded as the most general ``measurement'' that acts jointly on $k$ quantum channels, yielding a classical output. Crucially, and contrarily to parallel and sequential strategies, general strategies do not physically impose any particular order in which the copies of the channel must be applied. The testers that characterize these strategies, the general testers, are the most general sets of positive semidefinite operators $T^\text{GEN}=\{T^\text{GEN}_i\}_{i=1}^N, T^\text{GEN}_i\in\Lcal(\Hcal^{I}\otimes\Hcal^{O})$ that satisfy $p(i|C_1,\ldots,C_k)=\tr\left[T^\text{GEN}_i(C_1\otimes\ldots\otimes C_k)\right]$.
Let $W^\text{GEN}\coloneqq \sum_i T^\text{GEN}_i$. Then, general testers can be equivalently defined as
\begin{align}
    T^\text{GEN}_i &\geq 0 \ \ \ \forall \, i \\
    \tr[W^\text{GEN}(C_1 &\otimes\ldots\otimes C_k)] = 1, \label{eq::wnorm}
\end{align}
for all Choi states of quantum channels $C_i\in\Lcal(\Hcal^{I_i}\otimes\Hcal^{O_i})$. This is equivalent to defining $W^\text{GEN}$ to be a general $k$-partite process matrix.~\cite{oreshkov12,araujo15} For the cases of $k=2$ and $k=3$, we provide a characterization of $W^\text{GEN}$ in terms of linear constraints in Appendix~\ref{app:processmatrices}. 

Ordered processes, such as $W^\text{PAR}$ and $W^\text{SEQ}$, form a subset of the general processes $W^\text{GEN}$. However, some general processes are known not to respect the causal constraints of ordered processes, that is, they are neither parallel nor sequential.~\cite{oreshkov12} Moreover, they cannot be described as convex mixtures of ordered processes, in the bipartite case, or other more appropriate notions of mixtures of causal order, in the multipartite case.~\cite{wechs19} Such general processes have been termed to exhibit an \textit{indefinite causal order}~\cite{oreshkov12} and have been shown to bring advantages to several quantum information tasks, such as quantum computation,~\cite{araujo14}
communication complexity,~\cite{feix15,guerin16} discrimination of bipartite non-signaling channels,~\cite{chiribella12} violation of fully and semi-device-independent causal inequalities,~\cite{oreshkov12,branciard16,bavaresco19} inversion of unknown unitary operations,~\cite{quintino18} and, more recently, discrimination of general quantum channels.~\cite{bavaresco20}

Contrarily to the parallel and sequential cases, a realization of general testers in terms of quantum operations is an open problem. More specifically, a general tester can always be constructed from a general process matrix and a POVM,~\cite{bavaresco20} as illustrated in Fig.~\ref{fig::realization}(c). However, only a subset of process matrices are currently known to be realizable with quantum operations.~\cite{wechs21} This subset, known as ``coherent quantum control of causal orders'', has been shown to bring advantage to the discrimination of general channels.~\cite{bavaresco20} See Ref.~\cite{bavaresco20} for a more detailed discussion about the realization of testers.
\\

For any chosen strategy, the maximal probability of successful discrimination of an ensemble of $N$ channels $\Ecal=\{p_i,C_i\}_{i=1}^N$ using $k$ copies is given by
\begin{equation}\label{eq:SDP}
    P^\Scal \coloneqq \max_{\{T^\Scal_i\}}\sum_{i=1}^N p_i \tr\left(T^\Scal_i\,C_i^{\otimes k}\right),
\end{equation}
where $\Scal\in\{\text{PAR},\text{SEQ},\text{GEN}\}$. 

For all three classes of strategies, $P^\Scal$ can be computed via semidefinite programming (SDP). A short review of dual and primal SDP problems associated with the maximal probability of successful discrimination is given in Appendix~\ref{app:processmatrices}.

%%%%%%%%%%%%%%%%%%%%%%%%%%%%%%%%%%%%%%%%%%%%%%%%%%%%%%%%%%%%%%%%%
%%%%%%%%%%%%%%%%%%%%%%%%%%%%%%%%%%%%%%%%%%%%%%%%%%%%%%%%%%%%%%%%%

\section{Discrimination of unitary channels} 

In this section, we present our results concerning the discrimination of unitary channels. Here, unitary channels, also called unitary operations, will be simply denoted by unitary operators $U$ that satisfy $UU^\dagger=\id$, and these terms will be used interchangeably. The Choi operators of unitary channels will be denoted as $\dket{U}\dbra{U}\in\Lcal(\Hcal^I\otimes\Hcal^O)$, where 
\begin{equation}
    \dket{U}\coloneqq\sum_i(\id\otimes U)\ket{ii}.
\end{equation}

We start considering a discrimination task that involves unitary channels that form a group (see Appendix~\ref{app::proofunitarygroup}). We show that, for an ensemble composed of a set of unitary channels that forms a group and a uniform distribution, parallel strategies do not only perform as well as the sequential ones~\cite{chiribella08-1} but are, indeed, the optimal strategies for discrimination---even considering general strategies that may involve indefinite causal order.

\begin{theorem}\label{thm::unitarygroup}
    For ensembles composed of a uniform probability distribution and a set of unitary channels that forms a group up to a global phase, in discrimination tasks that allow for $k$ copies, parallel strategies are optimal, even when considering general strategies.
    
    More specifically, let $\mathcal{E}=\{p_i,U_i\}_i$ be an ensemble with $N$ unitary channels, where $p_i=\frac{1}{N}\ \forall\,i$ and the set $\{U_i\}_i$ forms a group up to a global phase. Then, for any number of copies $k$ and for every general tester $\{T^\text{GEN}_i\}$, there exists a parallel tester $\{T^\text{PAR}_i\}_i$ such that
    {\small{
    \begin{equation}
	    \frac{1}{N}\sum_{i=1}^N \tr\Big(T^\text{PAR}_i \dketbra{U_i}{U_i}^{\otimes k}\Big) 
	    =
	    \frac{1}{N}\sum_{i=1}^N \tr\Big(T^\text{GEN}_i \dketbra{U_i}{U_i}^{\otimes k}\Big).
    \end{equation}
    }}    
\end{theorem}    

The proof of this theorem can be found in Appendix~\ref{app::proofunitarygroup}. 
\\

Theorem~\ref{thm::unitarygroup} has two crucial hypotheses: (1) the set of unitary operators $\{U_i\}$ forms a group and (2) the distribution $\{p_i\}$ is uniform. If at least one of these hypotheses is not satisfied, then Theorem~\ref{thm::unitarygroup}, in fact, does not hold, as we show in the following.

\begin{theorem}\label{thm::seq_vs_par}
    There exist ensembles of unitary channels for which sequential strategies of discrimination outperform parallel strategies. Moreover, sequential strategies can achieve perfect discrimination in some scenarios where the maximal probability of success of parallel strategies is strictly less than one.
\end{theorem}   

Let us start with the case where the set of unitary channels does not form a group, but the probability distribution of the ensemble is uniform. In the following, $\sigma_x$, $\sigma_y$, and $\sigma_z$ are the Pauli operators and $H\coloneqq\ketbra{+}{0}+\ketbra{-}{1}$, where $\ket{\pm}\coloneqq \frac{1}{\sqrt{2}}(\ket{0}\pm \ket{1})$, is the Hadamard gate. Throughout the paper, we take $\sqrt{A}\coloneqq\sum_i\sqrt{\lambda_i}\ketbra{i}{i}$ to be the square root of an arbitrary diagonalizable operator $A=\sum_i\lambda_i\ketbra{i}{i}$.

\begin{example}\label{ex::k2N4}
    The ensemble composed of a uniform probability distribution and $N=4$ qubit unitary channels given by $\{U_i\} = \{\id,\sqrt{\sigma_x},\sqrt{\sigma_y},\sqrt{\sigma_z}\}$, in a discrimination task that allows for $k=2$ copies, can be discriminated under a sequential strategy with probability of success $P^\text{SEQ}=1$, while any parallel strategy yields $P^\text{PAR}<1$.
\end{example}

A straightforward sequential strategy that attains perfect discrimination of this ensemble can be constructed by first noting that $\sqrt{\sigma_i}\sqrt{\sigma_i}=\sigma_i$; hence, a simple composition of the unitary operators $U_i$ leads to the Pauli operators, which are perfectly discriminated with a bipartite maximally entangled state and a joint measurement in the Bell basis. The proof that any parallel strategy that applies two copies can never attain perfect discrimination is provided in Appendix~\ref{app::proofseq_vs_par} and applies the method of computer-assisted proofs developed in Ref.~\cite{bavaresco20}. 

Another example of this phenomenon is showed by the ensemble $\{U_i\} = \{\id,\sigma_x,\sigma_y,\sqrt{\sigma_z}\}$ with uniform probability distribution, which also satisfies $P^\text{PAR}<P^\text{SEQ}=1$ for a discrimination task with $k=2$ copies. In Appendix~\ref{app::proofseq_vs_par}, we show that such an ensemble can actually be discriminated perfectly by a sequential strategy that uses, on average, $1.5$ copies.

The next example concerns a set of unitary channels that forms a group, but the probability distribution of the ensemble is not uniform.

\begin{example}\label{ex::k2N8}
    Let $\{\id,\,\sigma_x,\,\sigma_y,\,\sigma_z,\,H,\,\sigma_xH,\,\sigma_yH,\,\sigma_zH\}=\{U_i\}$ be a tuple of $N=8$ unitary channels that forms a group up to a global phase, and let $\{p_i\}$ be a probability distribution in which each element $p_i$ is proportional to the $i$-th digit of the number $\pi\approx3.1415926$, that is, $\{p_i\}=\{\frac{3}{31},\frac{1}{31},\frac{4}{31},\ldots,\frac{6}{31}\}$. For the ensemble $\{p_i,U_i\}$, in a discrimination task that allows for $k=2$ copies, sequential strategies outperform parallel strategies, i.e., $P^\text{PAR}<P^\text{SEQ}$.
\end{example}
The proof of this example is also in Appendix~\ref{app::proofseq_vs_par}, and applies the same method of computer-assisted proofs. In Example~\ref{ex::k2N8}, we have set the distribution $\{p_i\}$ to be proportional to the $i^\text{th}$ digit of the constant $\pi$ to emphasize that the phenomenon of sequential strategies outperforming the parallel ones when the set of unitary channels forms a group does not require a particularly well chosen non-uniform distribution. In practice, we have observed that with randomly generated distributions, optimal strategies often respect $P^\text{PAR}<P^\text{SEQ}$.

In both of the aforementioned examples, general strategies do not outperform sequential strategies. However, for the case of discrimination of unitary channels using $k=3$ copies, we show that general strategies are, indeed, advantageous.

\begin{theorem}\label{thm::gen_vs_seq_vs_par}
    There exist ensembles of unitary channels for which general strategies of discrimination outperform sequential strategies.
\end{theorem}

Let us start again with the case where the set of unitary channels does not form a group, but the probability distribution of the ensemble is uniform. For the following, we define $H_y\coloneqq\ketbra{+_y}{0}+\ketbra{-_y}{1}$, where $\ket{\pm_y}\coloneqq \frac{1}{\sqrt{2}}(\ket{0}\pm i\ket{1})$, and $H_P\coloneqq\ketbra{+_P}{0}+\ketbra{-_P}{1}$, where $\ket{+_P}\coloneqq \frac{1}{5}(3\ket{0} + 4\ket{1})$ and $\ket{-_P}\coloneqq \frac{1}{5}(4\ket{0} - 3\ket{1})$.

\begin{figure*}%[h!]
\begin{center}
	\includegraphics[width=\textwidth]{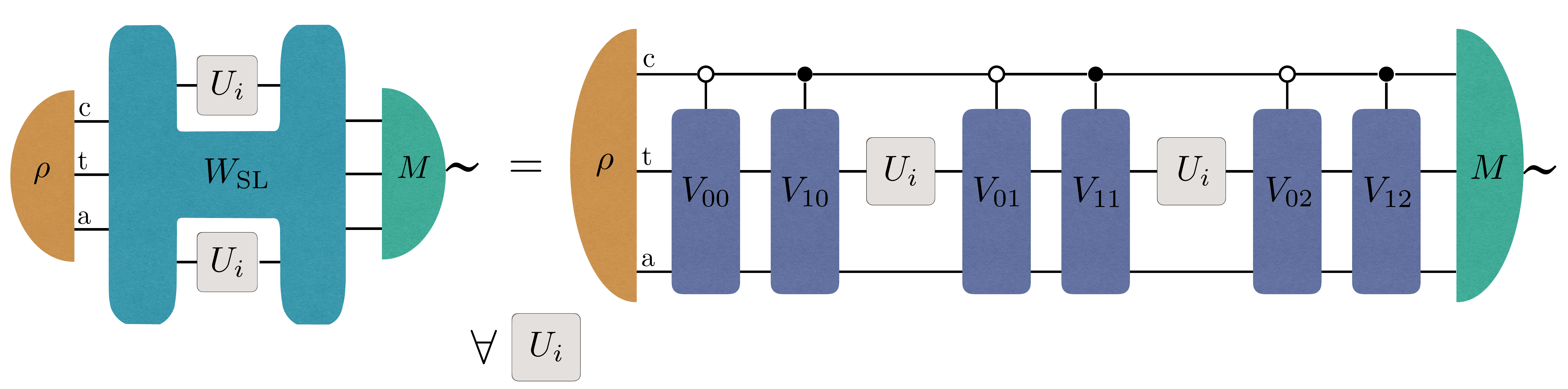}
	\caption{Illustration of a $2$-copy sequential strategy that attains the same probability of successful discrimination as any $2$-copy switch-like strategy, for all sets of unitary channels $\{U_i\}_{i=1}^N$. Line ``c'' represents a control system, ``t'' represents a target system, and ``a'' represents an auxiliary system. Both strategies can be straightforwardly extended to $k$ copies (see Appendix~\ref{app::proofswitchlike}).}
\label{fig:SL}
\end{center}
\end{figure*}

\begin{example}\label{ex::k3N4}
    For the ensemble composed of a uniform probability distribution and $N=4$ unitary channels given by $\{U_i\} = \{\sqrt{\sigma_x},\sqrt{\sigma_z},\sqrt{H_y},\sqrt{H_P}\}$, in a discrimination task that allows for $k=3$ copies, general strategies outperform sequential strategies and sequential strategies outperform parallel strategies. Therefore, the maximal probabilities of success satisfy the strict hierarchy $P^\text{PAR}<P^\text{SEQ}<P^\text{GEN}$.
\end{example}
The proof of this example can be found in Appendix~\ref{app::proofgen_vs_seq_vs_par}. 
\\

General strategies can also be advantageous for the discrimination of an ensemble composed of a non-uniform probability distribution and a set of unitary channels that forms a group. Let the set of unitary operators in Example~\ref{ex::k3N4} be the set of generators of a group (potentially with an infinite number of elements). Now, consider the ensemble composed of such a group and a probability distribution given by $p_i=\frac{1}{4}$ for the four values of $i$ corresponding to the four unitary operators which are the generators of the group, and $p_i=0$ otherwise. It is straightforward to see that the maximal probabilities of successfully discriminating this ensemble would be the same as the ones in Example~\ref{ex::k3N4}, hence satisfying $P^\text{PAR}<P^\text{SEQ}<P^\text{GEN}$. Although somewhat artificial, this example shows that advantages of general strategies are, indeed, possible for this kind of unitary channel ensemble. 
\\

Although general indefinite-causal-order strategies can be advantageous for the discrimination of unitary channels, this is not the case for one particular sub-class of general strategies: those which can be constructed from the quantum switch.~\cite{chiribella13-01} 

Let $V_{mn}$, with $m\in\{0,1\}$, $n\in\{0,1,2\}$, be unitary operators that act on a target and an auxiliary system and $U_1,U_2$ be unitary operators that act only on the target system. Finally, let $\{\ketbra{m}{m}^c\}_m$ be projectors that act on a control system. Then, we define the \textit{switch-like} superchannel, which transforms a pair of unitary channels into one unitary channel, according to
\begin{align}
\begin{split}\label{eq::slsc}
    \Wcal_\text{SL}(U_1,U_2)\coloneqq &\ketbra{0}{0}^c \otimes V_{02}\,({U_2}\otimes\id)\,V_{01}\,({U_1}\otimes\id)\,V_{00} \\
    + &\ketbra{1}{1}^c\otimes V_{12}\,({U_1}\otimes\id)\,V_{11}\,({U_2}\otimes\id)\,V_{10},
\end{split}
\end{align}
where $\id$ is the identity operator acting on the auxiliary system. In the case where $V_{mn}=\id\ \forall\,m,n$, one recovers the standard quantum switch.~\cite{chiribella13-01} 
The switch-like superchannel has been previously considered in Refs.~\cite{yokojima20,barrett20} in the context of reversability-preserving transformations. 
Generalizations of the switch-like superchannel that transform $k$ instead of $2$ unitaries are presented in detail in Appendix~\ref{app::proofswitchlike}, applying unitaries $\{V_{mn}\}_{m,n}$, with $m\in\{0,\ldots,k!-1\}$ and $n\in\{0,\ldots,k\}$, and considering all permutations of the target unitaries $\{U_l\}_{l=1}^k$. Such $k$-slots switch-like superchannels have been shown to be implementable via coherent quantum control of causal orders.~\cite{wechs21}

Now, let $W^\text{SL}\in\L(\H^P\otimes\H^I\otimes\H^O\otimes\H^F)$ be the $k$-slot switch-like process associated with the $k$-slot generalization of the switch-like superchannel in Eq.~\eqref{eq::slsc}. A general discrimination strategy, given by the $k$-copy switch-like tester $T^\text{SL}=\{T^\text{SL}_i\}$, $T^\text{SL}_i\in\Lcal(\Hcal^{I}\otimes\Hcal^{O})$ can be constructed using the $k$-slot switch-like process $W^\text{SL}$, a quantum state $\rho\in\Lcal(\Hcal^P)$ that acts on the ``past'' space of the $k$ slots of $W^\text{SL}$, and a POVM $\{M_i\},M_i\in\Lcal(\Hcal^F)$, that acts on the ``future'' space, according to
\begin{equation}
    T^\text{SL}_i\coloneqq \tr_{PF}\Big[(\rho\otimes\id)W^\text{SL}(\id\otimes M_i)\Big]
\end{equation}
where the identity operators $\id$ act on the correspondent complementary spaces.

We show that such switch-like strategies exhibit no advantage over sequential strategies for the discrimination of $N$ unitary channels using $k$ copies.

\begin{theorem}\label{thm::switchlike}
The action of the switch-like process on $k$ copies of a unitary channel can be equivalently described by a sequential process that acts on $k$ copies of the same unitary channel.
    
Consequently, in a discrimination task involving the ensemble $\mathcal{E}=\{p_i,U_i\}_i$ composed of $N$ unitary channels and some probability distribution, and that allows for $k$ copies, for every switch-like tester $\{T^\text{SL}_i\}_i$, there exists a sequential tester $\{T^\text{SEQ}_i\}_i$ that attains the same probability of success, according to
    \begin{equation}
	    \sum_{i=1}^N p_i \tr\Big(T^\text{SL}_i \dketbra{U_i}{U_i}^{\otimes k}\Big) 
    	=
	    \sum_{i=1}^N p_i \tr\Big(T^\text{SEQ}_i \dketbra{U_i}{U_i}^{\otimes k}\Big).
	\end{equation}
\end{theorem}

The proof can be found in Appendix~\ref{app::proofswitchlike}, where we provide a simple construction of a sequential strategy that performs as well as any switch-like strategy using the same number of copies for unitary channel discrimination. A graphical representation of such a strategy in the case of $k=2$ copies is represented in Fig.~\ref{fig:SL}.

%%%%%%%%%%%%%%%%%%%%%%%%%%%%%%%%%%%%%%%%%%%%%%%%%%%%%%%%%%%%%%%%%
%%%%%%%%%%%%%%%%%%%%%%%%%%%%%%%%%%%%%%%%%%%%%%%%%%%%%%%%%%%%%%%%%

\section{Uniformly sampled unitary channels} \label{sec:numeric}
\begin{table}[t!]
\begin{center}
{\renewcommand{\arraystretch}{1.5}
\begin{tabular}{| c | c | c |}               
        \multicolumn{3}{c}{Uniformly sampling qubit unitary channels} \\ 
	    \hline 
		$N$ & $k=2$ & $k=3$ \\    
		\hline 
		\ \ $2$ \ \ & $\ \boldsymbol{P^\text{\textbf{PAR}}=P^\text{\textbf{SEQ}}}=P^\text{GEN} \ $ & $\ \boldsymbol{P^\text{\textbf{PAR}}=P^\text{\textbf{SEQ}}}=P^\text{GEN} \ $ \\			
		        $3$         &     $P^\text{PAR}=P^\text{SEQ}=P^\text{GEN}$         &  $P^\text{PAR}<P^\text{SEQ}=P^\text{GEN}$ \\	
		        $4$         &     $P^\text{PAR}<P^\text{SEQ}=P^\text{GEN}$         &  $P^\text{PAR}<P^\text{SEQ}<P^\text{GEN}$ \\	
		        $\vdots$    &     $\vdots$                                         &  $\vdots$                                 \\
		        $10$         &     $P^\text{PAR}<P^\text{SEQ}=P^\text{GEN}$         &  $P^\text{PAR}<P^\text{SEQ}<P^\text{GEN}$ \\
		        $\vdots$    &     $\vdots$                                         &  $\vdots$ \\
		        $25$        &     $P^\text{PAR}\approx P^\text{SEQ}=P^\text{GEN}$  &  $P^\text{PAR}< P^\text{SEQ} \approx P^\text{GEN}$\\
        \hline 
\end{tabular}
}
\end{center}
\caption{Summary of numerical findings: Gaps between different strategies of discrimination using $k$ copies of ensembles of $N$ uniformly distributed qubit unitary channels sampled according to the Haar measure.
The bold equalities on row $N=2$ denote known analytical results (see Ref.~\cite{chiribella08-1}). A strict inequality indicates that examples of ensembles that exhibit such gaps were encountered. An equality indicates that, for all sampled ensembles, no gap was encountered, up to numerical precision. The number of sampled sets ranged from $500$ to $50\,000$.} 
\label{tbl::table}
\end{table}

\begin{figure}%[h!]
\begin{center}
	\includegraphics[width=\columnwidth]{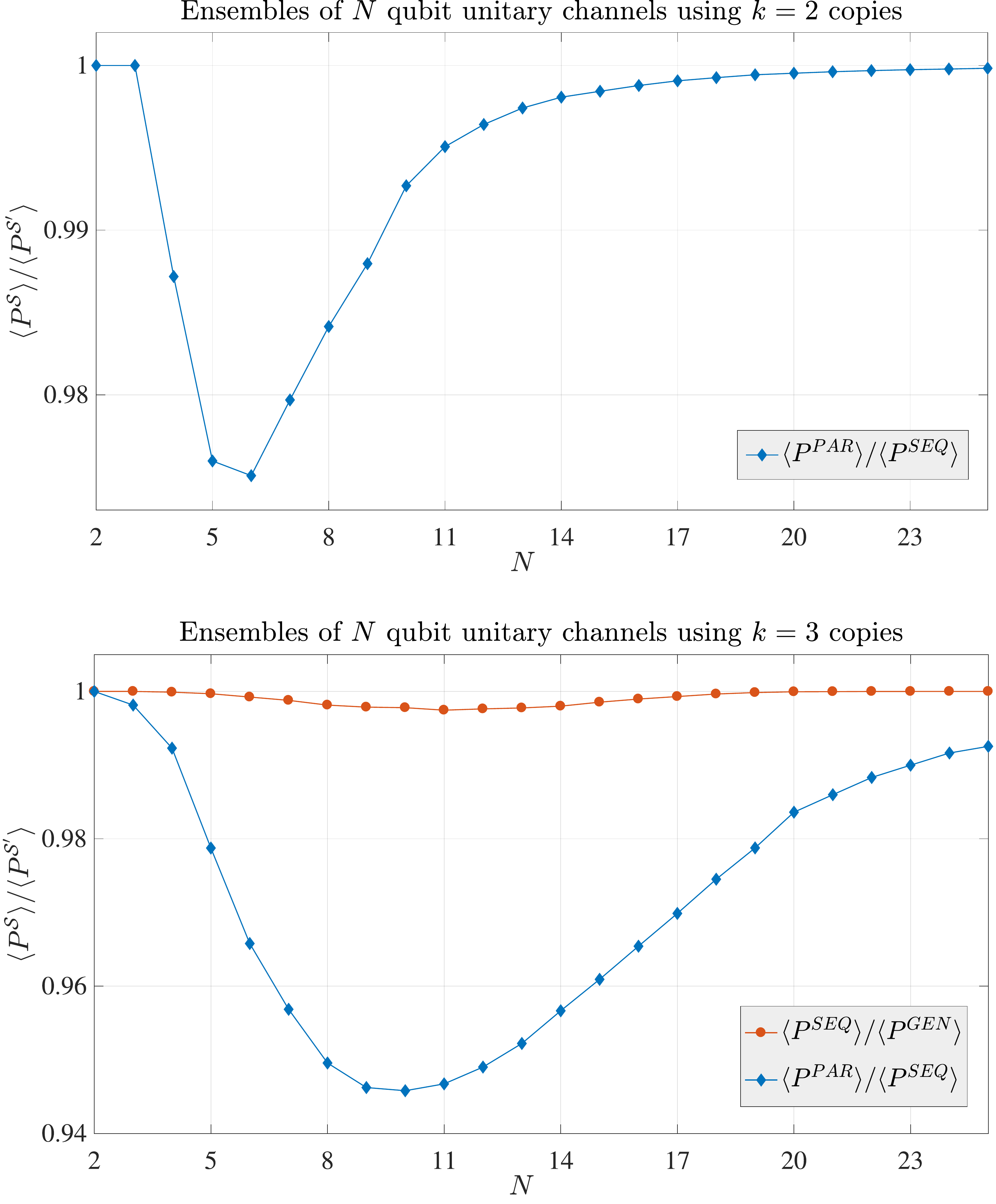}
	\caption{Ratios of the averages of the maximal probability of successful discrimination of ensembles of $N\in\{2,\ldots,25\}$ uniformly distributed qubit unitary channels sampled according to the Haar measure. 
	Top: For $k=2$ copies, the ratio between parallel and sequential strategies $\avg{P^\text{PAR}}/\avg{P^\text{SEQ}}$. To compute the averages, $1\,000$ sets were sampled for each $N$. 
	Bottom: For $k=3$ copies, the ratio between parallel and sequential strategies $\avg{P^\text{PAR}}/\avg{P^\text{SEQ}}$ and between sequential and general strategies $\avg{P^\text{SEQ}}/\avg{P^\text{GEN}}$. To compute the averages, $500$ sets were sampled for each $N$.}
\label{fig::plots_ratios}
\end{center}
\end{figure}

The advantage of sequential and general strategies in the discrimination of unitary channels is not restricted to the main examples given for Theorems~\ref{thm::seq_vs_par} and~\ref{thm::gen_vs_seq_vs_par}. In fact, by sampling sets of unitary operators uniformly distributed according to the Haar measure, and using these sets to construct ensembles with uniform probability distribution, one can find several other examples of the advantage of sequential and general strategies. A summary of our numerical findings is presented in Table~\ref{tbl::table}.

\textit{Qubits, $k=2$}.
In the scenario of qubit unitary channels and $k=2$ copies, we have observed gaps between the performance of parallel and sequential strategies for ensembles of $N\in\{4,\ldots,25\}$ unitary channels. For the case of $N=2$, it is known that such gaps do not exit~\cite{chiribella08-1}, and for the case of $N=3$, no gap was discovered. By calculating the averages of the maximal probabilities of success, we observed that the minimum ratio $\avg{P^\text{PAR}}/\avg{P^\text{SEQ}}$ occurred at $N=6$. At $N=25$, gaps were hardly detected. This behavior can be visualized in the plot of Fig.~\ref{fig::plots_ratios} (top). 

\textit{Qutrits, $k=2$}.
For the case of qutrit unitary channels and $k=2$ copies, we discovered a gap between the performance of parallel and sequential strategies already for a discrimination task of only $N=3$ unitary channels, while in the qubit case, the first example of this phenomenon was found only for $N=4$. In this scenario as well, no advantages of general strategies were found. These results are not plotted.

\textit{Qubits, $k=3$}.
In the scenario of qubit unitary channels and $k=3$ copies, the advantage of general strategies over causally ordered ones (parallel and sequential) is common. Still considering uniformly sampled qubit unitary channels, in the $3$-copy case, we have found a strict hierarchy of discrimination strategies in scenarios of $N\in\{4,\ldots,19\}$. For $N\in\{20,\ldots,25\}$, the advantage of sequential strategies was clear but that of general strategies was hardly detected. Also in the case of $N=3$, only an advantage of sequential over parallel strategies was found. The minimum ratios of the averages $\avg{P^\text{PAR}}/\avg{P^\text{SEQ}}$ and $\avg{P^\text{SEQ}}/\avg{P^\text{GEN}}$ were found at $N=10$ and $N=11$, respectively. These results are plotted in Fig.~\ref{fig::plots_ratios} (bottom).

The number of unitary channel ensembles sampled to calculate the average values of the maximal probabilities of success ranged from $50\,000$ for $d=2$, $k=2$, $N=2$ to $100$ for $d=3$, $k=2$, $N=3$.

In Fig.~\ref{fig::plots_ratios}, one can see how the ratios between the maximal probabilities of success of different classes of strategies decrease as the number $N$ of uniformly sampled unitary channels being discriminated increases. This observation is in line with the idea that, in the limit where the ensemble is composed of all qubit unitary channels, therefore forming the group $SU(2)$, it is expected that parallel strategies would be optimal. In Sec.~\ref{sec::bounds}, we formally analyze the asymptotic behavior of $\avg{P^\Scal}$, while providing an absolute upper bound for the maximal probability of success under any strategy.

%%%%%%%%%%%%%%%%%%%%%%%%%%%%%%%%%%%%%%%%%%%%%%%%%%%%%%%%%%%%
%%%%%%%%%%%%%%%%%%%%%%%%%%%%%%%%%%%%%%%%%%%%%%%%%%%%%%%%%%%%

\section{Ultimate upper bound for $P^\Scal$ for any set of unitary channels}\label{sec::bounds}

We now present an upper bound for the maximal probability of success for discriminating a set of $N$ $d$-dimensional unitary channels with general strategies when $k$ copies are available. Our result applies to \textit{any} ensemble of unitary channels $\mathcal{E}=\{p_i,U_i\}_{i=1}^N$ where $p_i=\frac{1}{N}$ is a uniform probability distribution. Since general testers are the most general strategies that are consistent with a channel discrimination task, our result constitutes an ultimate upper bound for discriminating uniformly distributed unitary channels. Additionally, as we show later, our bound can be saturated by particular choices of unitary channels.

\begin{theorem}[Upper bound for general strategies] \label{thm::upper_bound}
Let $\mathcal{E}=\{p_i,U_i\}_{i=1}^N$ be an ensemble composed of $N$ $d$-dimensional unitary channels and a uniform probability distribution. The maximal probability of successful discrimination of a general strategy with $k$ copies is upper bounded by
\begin{equation}
    P^\text{GEN}\leq \frac{1}{N}\,\gamma(d,k),
\end{equation}
where $\gamma(d,k)$ is given by
\begin{equation}
   \gamma(d,k) \coloneqq 
{k+d^2-1\choose k}= \frac{(k+d^2-1)!}{k!(d^2-1)!}.
\end{equation}
\end{theorem}

The proof of this theorem is in Appendix~\ref{app::upperbound}.

\begin{figure}%[h!]
\begin{center}
	\includegraphics[width=\columnwidth]{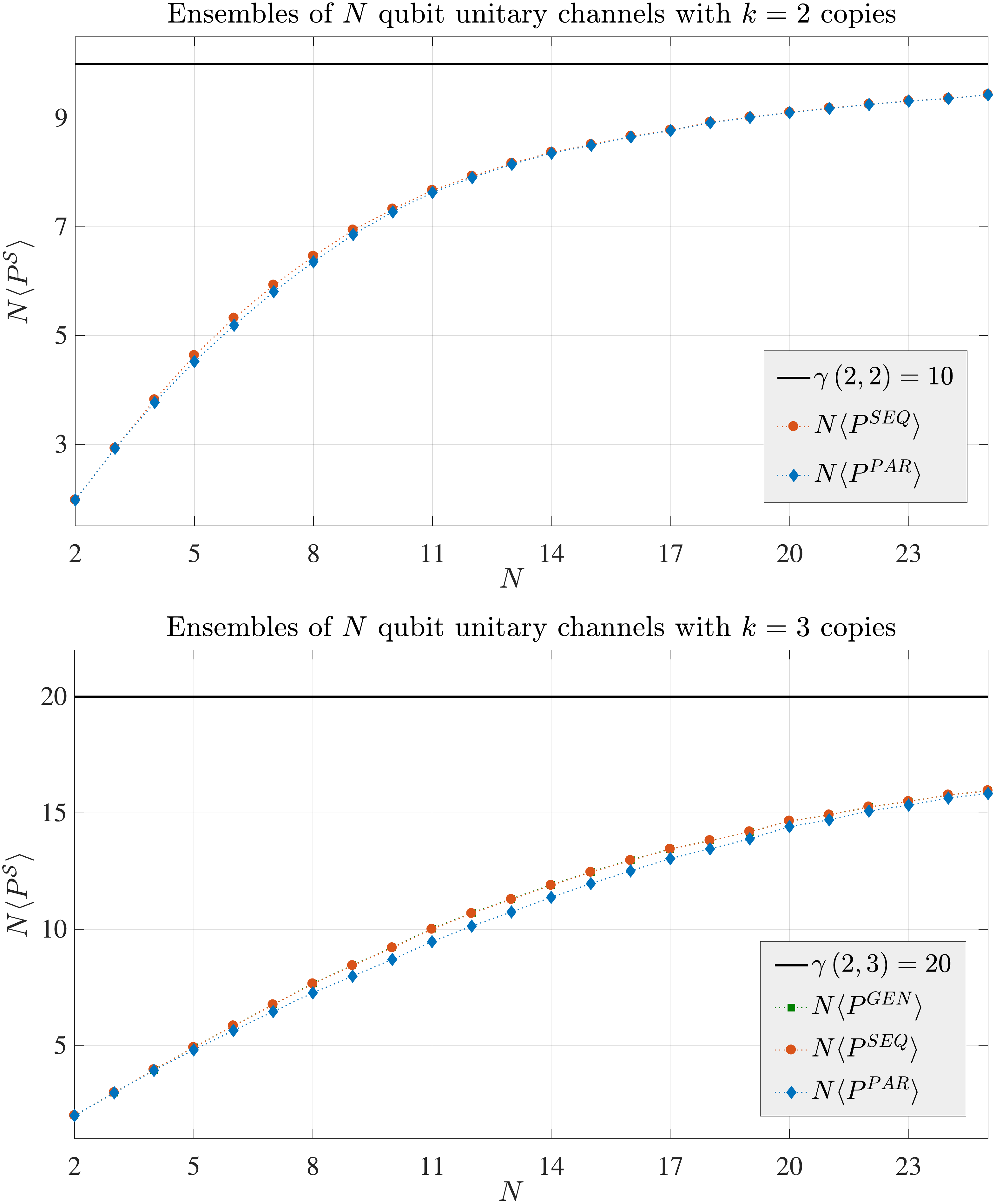}
	\caption{Upper bound and average of the maximal probability of successful discrimination of sets of $N\in\{2,\ldots,25\}$ qubit unitary operators sampled according to the Haar measure.
	The plotted curves correspond to the average $\avg{P^\Scal}$ multiplied by $N$ and the constant line is the analytical upper bound of $\gamma(2,k)$ derived in Thm.~\ref{thm::upper_bound}. Note how all $N\avg{P^\Scal}$ asymptotically approach the upper bound $\gamma(2,k)$ as $N$ increases, although somewhat slowly.
	Top: For $k=2$ copies, $N\avg{P^\text{PAR}}$ and $N\avg{P^\text{SEQ}}$ are plotted against the upper bound of $\gamma(2,2)=10$. To compute the averages, $1\,000$ sets were sampled for each $N$.
	Bottom: For $k=3$ copies, $N\avg{P^\text{PAR}}$, $N\avg{P^\text{SEQ}}$, and $N\avg{P^\text{GEN}}$ are plotted against the upper bound of $\gamma(2,3)=20$. The difference between plotted points of $N\avg{P^\text{SEQ}}$ and $N\avg{P^\text{GEN}}$ is not visible. To compute the averages, $500$ sets were sampled for each $N$.}
\label{fig::plots_limit}
\end{center}
\end{figure}

The upper bound given by Thm.~\ref{thm::upper_bound} can be attained by some particular choices of unitary channels that form a group (and hence, by Thm.~\ref{thm::gen_vs_seq_vs_par}, attainable by a parallel strategy). In particular, it follows from Refs.~\cite{chiribella04,hayashi05,chiribella06} that when sets of unitary operators $\{U_i\}_i$ form group $k$-designs, there exists a parallel strategy such that
\begin{align}
     P^\text{PAR}=\frac{1}{N}\,\gamma(d,k).
\end{align}

A set of unitary operators $\{ U_i\}_{i=1}^N$ forms a group $k$-design when the set forms a group and respects the relation
\begin{equation}\label{eq::kdesign}
    \int_\text{Haar} U^{\otimes k}\,\rho\ {U^{\otimes k}}^\dagger \text{d}U =
    \frac{1}{N}\sum_{i=1}^N U_i^{\otimes k}\,\rho\ {U_i^{\otimes k}}^\dagger
\end{equation}
for every linear operator $\rho$.

In Sec.~\ref{sec:numeric}, we analyzed the problem of discriminating $N$ unitary operators uniformly sampled according to the Haar measure. For any fixed dimension $d$ and number of copies $k$, if we sample a very large number $N$ of unitary operators $(N\to\infty)$, when the distribution $\{p_i\}_i$ is uniform, the maximal probability of success will be very small $(P^\Scal\to0)$. At the same time, due to the uniform properties of the Haar measure, if we uniformly sample a large number $N$ of unitary operators, we obtain a set that approaches the group of all unitary operators of the given dimension, which is equivalent to the group $SU(d)$, approximately satisfying the group $k$-design condition in Eq.~\eqref{eq::kdesign}. Therefore, for large $N$, we have the asymptotic behavior
\begin{equation}
    N \avg{P^\text{PAR}} \approx  N \avg{P^\text{SEQ}} \approx   N  \avg{P^\text{GEN}}\approx \gamma(d,k).
\end{equation}
This argument can be made more rigorous by recognizing that $SU(d)$ is a compact Lie group, treating probabilities as probability densities, and maximizing likelihood.~\cite{chiribella04,chiribella06}

As an example from our numerical calculations, in Fig.~\ref{fig::plots_limit}, for the case of $d=2$ and $k=2,3$, we can visualize the behavior of $N\avg{P^\Scal}$ as a function of $N$ and how it asymptotically approaches the constant $\gamma(d,k)$, which takes the values of $\gamma(2,2)=10$ and $\gamma(2,3)=20$, with increasing $N$.

%%%%%%%%%%%%%%%%%%%%%%%%%%%%%%%%%%%%%%%%%%%%%%%%%%%%%%%%%%%%%%%%%
%%%%%%%%%%%%%%%%%%%%%%%%%%%%%%%%%%%%%%%%%%%%%%%%%%%%%%%%%%%%%%%%%

\section{Conclusion}

We extended the unified tester formalism of Ref.~\cite{bavaresco20} to the case of $k$ copies and applied it particularly to the study of unitary channel discrimination. Our first contribution was to prove that, in a discrimination task of an ensemble of a set of unitary channels that forms a group and a uniform probability distribution, parallel strategies are always optimal, even when compared against the performance of general strategies.

Subsequently, we showed an example of a unitary channel discrimination task in which a sequential strategy outperforms any parallel strategy. Our result consists of the first demonstration of such a phenomenon, to the extent of our knowledge. Our task involves two copies and four unitary channels, which can be perfectly discriminated with a sequential strategy but not with a parallel one. We explicitly provided the optimal discrimination strategy for this task. 

We also showed that general strategies that involve indefinite causal order are advantageous for the discrimination of unitary channels. Our simplest example of this phenomenon is a task of discriminating among four unitary channels using three copies. While our optimal parallel and sequential strategies can be straightforwardly implemented with (ordered) quantum circuits, a potential quantum realization of the optimal general strategies that are advantageous in this scenario remains an open problem. We then demonstrated that general strategies that are created from switch-like transformations, which are known to be, in principle, implementable with quantum operations,~\cite{wechs21} can never perform better than sequential strategies for unitary channel discrimination. 

The final result of our work was an upper bound for the maximal probability of success of any set of unitary channels using any number of copies under any strategy. This ultimate bound applies to any possible discrimination strategy and was shown to be tight, attained by discrimination tasks of unitary groups that form a $k$-design.

The concrete examples of the advantages of sequential and general strategies in this work focused on discrimination tasks that use $k=2$ or $3$ copies. An open question of our work is how these strategy gaps would scale with larger values of $k$. The preliminary results presented here indicate that the advantage of sequential over parallel strategies and of general over sequential strategies should be even more accentuated as a higher number of copies are allowed. This idea is supported by the intuition that the number of different ways in which one can construct sequential strategies, as compared to parallel strategies, increases with the number of slots $k$. Similarly, we expect such a phenomenon to exist for the general case. It would then be interesting to find out exactly the rate with which these gaps increase with $k$.

No advantage of general strategies was found in scenarios involving discrimination of unitary channels using only $k=2$ copies. We conjecture that, when considering $k=2$ copies, such an advantage is, indeed, not possible for any number $N$ of unitary channels. We also remark that, when considering $k=2$ copies, Refs.~\cite{yokojima20,barrett20} proved that superchannels that preserve reversibility (i.e. transform unitary channels into unitary channels) are necessarily of the switch-like form. Intuitively, it seems plausible that the optimal general strategy for discriminating unitary channels would be one that transforms unitary channels into unitary channels. This argument of reversibility preservation combined with our Theorem~\ref{thm::switchlike} might lead to a proof of our conjecture. 

Furthermore, we also conjecture that, when considering $N=2$ unitary channels, general strategies are not advantageous for any number of copies $k$. In this scenario, it has been proven that sequential strategies cannot outperform parallel ones,~\cite{chiribella08-1} and we believe this to also be the case for general strategies. The task of discriminating between two unitary channels can always be recast as the problem of discriminating a unitary operator from the identity operator. In the parallel case, the probability of successful discrimination has been shown to be related to the spread of the eigenvalues of this unitary operator.~\cite{acin01,dariano01} The proof presented in Ref.~\cite{chiribella08-1} explores how sequential strategies affect the spread of the eigenvalues of unitary operators, to conclude that they cannot outperform the parallel ones. A better understanding of how general strategies affect the spread of the eigenvalues of unitary operators could lead to a conclusive answer for this conjecture.

Finally, an interesting open question that could follow our work is how quantum-specific are the phenomena presented here. For instance, Ref.~\cite{harrow10} showed that adaptive strategies may outperform parallel ones in classical channel discrimination; however, their examples did not concern a classical analog of unitary channels. Assuming the analog of a unitary channel in classical information theory to be a classical channel that maps deterministic probability distributions into deterministic probability distributions, i.e., channels that can be described by a permutation matrix, it would be interesting to investigate whether their discrimination could also be enhanced by sequential strategies. Additionally, Refs.~\cite{baumeler16,araujo16} showed that indefinite causal order also manifests itself in classical processes. Hence, an interesting question would be whether indefinite causality is also a useful resource for classical channel discrimination.

\begin{acknowledgments}
\noindent\textit{Acknowledgments.} The authors thank Cyril Branciard and an anonymous referee for the conference TQC~2022 for very useful comments on the first version of the manuscript. The authors are also grateful to Princess Bubblegum for performing demanding calculations.
J.B. acknowledges funding from the Austrian Science Fund (FWF) through Zukunftskolleg ZK03 and START Project Y879-N27.
M.M. was supported by the MEXT Quantum Leap Flagship Program (MEXT Q-LEAP) under Grant No. JPMXS0118069605 and by the Japan Society for the Promotion of Science (JSPS) through KAKENHI Grants Nos. 17H01694, 18H04286, and 21H03394. 
M.T.Q. acknowledges the Austrian Science Fund (FWF) through the SFB project BeyondC (sub-project F7103), a grant from the Foundational Questions Institute (FQXi) as part of the  Quantum Information Structure of Spacetime (QISS) Project (qiss.fr). The opinions expressed in this publication are those of the authors and do not necessarily reflect the views of the John Templeton Foundation. This project  received funding from the European Union’s Horizon 2020 Research and Innovation Programme under Marie Sk\l{}odowska-Curie Grant Agreement No. 801110 and the Austrian Federal Ministry of Education, Science and Research (BMBWF). It reflects only the authors' view and the EU Agency is not responsible for any use that may be made of the information it contains.
\end{acknowledgments}

\section*{Author declarations}
The authors have no conflicts to disclose.

\section*{Data Availability Statement}
Data sharing is not applicable to this article as no new data were created or analyzed in this study. 
\\

The code developed for this study is openly available at our online repository.~\cite{githubMTQ2}

%%%%%%%%%%%%%%%%%%%%%%%%%%%%%%%%%%%%%%%%%%%%%%%%%%%%%%%%
%%%%%%%%%%%%%%%%%%%%%%%%%%%%%%%%%%%%%%%%%%%%%%%%%%%%%%%%

\onecolumngrid
\appendix

\section*{Appendix}

The appendix is composed of the following sections: Appendix~\ref{app:processmatrices} presents the definition of $2$- and $3$-slot process matrices in terms of positivity and linear normalization constraints and a small review of the formulation of the maximal probability of successful discrimination in terms of primal and dual SDP problems; Appendix~\ref{app::proofunitarygroup} presents the proof of Theorem~\ref{thm::unitarygroup};
Appendix~\ref{app::proofseq_vs_par} presents the proof of Examples~\ref{ex::k2N4} and~\ref{ex::k2N8};
Appendix~\ref{app::proofgen_vs_seq_vs_par} presents the proof of Example~\ref{ex::k3N4}, %and~\ref{ex::k3N8};
Appendix~\ref{app::proofswitchlike} presents the proof of Theorem~\ref{thm::switchlike}; and, finally,
Appendix~\ref{app::upperbound} contains the proof of Theorem~\ref{thm::upper_bound}.
\\

Some of the sections in the Appendix will make use of the \textit{link product}~\cite{chiribella07} between two linear operators, which is a useful mathematical tool to compose linear maps that are represented by their Choi operators. If $\map{C} \coloneqq \map{B}\circ \map{A}$ is the composition of the linear maps  $\map{A}:\L(\H^1)\to\L(\H^2)$ and  $\map{B}:\L(\H^2)\to\L(\H^3)$, the Choi operator of $\map{C}$ is given by $C=A*B$, where $A$ and $B$ are the Choi operators of $\map{A}$ and $\map{B}$, respectively, and $*$ stands for the link product, which we now define.
Let $A\in\L(\H^1\otimes\H^2)$ and $B\in\L(\H^2\otimes\H^3)$ be linear operators. The link product $A*B\in\L(\H^1\otimes\H^3)$ is defined as
\begin{equation}
    A^{12}*B^{23}\coloneqq \tr_2\left[\Big((A^{12})^{T_2} \otimes \id^3\Big)\, \Big(\id^1 \otimes B^{23}\Big)\right], 
\end{equation}
where $(\cdot)^{T_2}$ stands for the partial transposition on the linear space $\H^2$.

We remark that identifying the linear spaces where the operators act is an important part of the link product, also, if we keep track on these linear spaces, the link product is commutative and associative.

%%%%%%%%%%%%%%%%%%%%%%%%%%%%%%%%%%%%%%%%%%%%%%%%%%%%%%%%%%%%%%%%%%%%%%%%%%%%%%%%%%%%%%%%%%%%%%%%%%%%%%%%%%%%
%%%%%%%%%%%%%%%%%%%%%%%%%%%%%%%%%%%%%%%%%%%%%%%%%%%%%%%%%%%%%%%%%%%%%%%%%%%%%%%%%%%%%%%%%%%%%%%%%%%%%%%%%%%%

\section{Semidefinite programming formulation and linear constraints for general processes of $k=2$ and $k=3$ slots}\label{app:processmatrices}

For sake of completeness, in the following, we present a characterization of general processes with $k=2$ and $k=3$ slots, which are bipartite and tripartite process matrices, in terms of positivity constraints and linear normalization constraints. While positivity is a consequence of physical considerations, the linear normalization constraints are a direct implication of Eq.~\eqref{eq::wnorm}. %We refer to Ref.~\cite{araujo15} for a detailed derivation.
We refer to Ref.~\cite{araujo15} for a detailed explanation of the method and Ref.~\cite{quintino19} for the explicit characterization of general processes with $k=3$ slots.

We recall that we represent the trace-and-replace operation on $\Hcal^X$ as $_{X}(\cdot)=\tr_X(\cdot)\otimes\frac{\id^X}{d_X}$, where $d_X=\text{dim}(\Hcal^X)$.

Let $W\in\Lcal(\Hcal^{A_IA_O}\otimes\Hcal^{B_IB_O})$ be a $2$-slot (bipartite) process matrix. Then,
\begin{align}
    W&\geq0 \\
    \tr(W) &= d_{A_O}d_{B_O} \\
    _{A_IA_O}W &= _{A_IA_OB_O}W \\
    _{B_IB_O}W &= _{A_OB_IB_O}W \\
    W + _{A_OB_O}W &= _{A_O}W + _{B_O}W. 
\end{align}

Now let $W\in\Lcal(\Hcal^{A_IA_O}\otimes\Hcal^{B_IB_O}\otimes\Hcal^{C_IC_O})$ be a $3$-slot (tripartite) process matrix. Then,
\begin{align}
    W&\geq0 \\
    \tr(W) &= d_{A_O}d_{B_O}d_{C_O} \\
    _{A_IA_OB_IB_O}W &= _{A_IA_OB_IB_OC_O}W \\
    _{B_IB_OC_IC_O}W &= _{A_OB_IB_OC_IC_O}W \\
    _{A_IA_OC_IC_O}W &= _{A_IA_OB_OC_IC_O}W \\
    _{A_IA_O}W + _{A_IA_OB_OC_O}W  &= _{A_IA_OB_O}W + _{A_IA_OC_O}W \\
    _{B_IB_O}W + _{A_OB_IB_OC_O}W  &= _{A_OB_IB_O}W + _{B_IB_OC_O}W \\
    _{C_IC_O}W + _{A_OB_OC_IC_O}W  &= _{A_OC_IC_O}W + _{B_OC_IC_O}W \\
    W + _{A_OB_O}W + _{A_OC_O}W + _{B_OC_O}W &= _{A_O}W + _{B_O}W + _{C_O}W + _{A_OB_OC_O}W. 
\end{align}

We also recall that the maximal probability of successful discrimination in Eq.~\eqref{eq:SDP} can be equivalently expressed by the primal SDP problem
\begin{flalign}\label{sdp::primal}
\begin{aligned}
    \textbf{given}\ \      &\{p_i,C_i\} \\
    \textbf{maximize}\ \   &\sum_i p_i \tr\left(T^\Scal_i\,C_i^{\otimes k}\right) \\ 
    \textbf{subject to}\ \ &T^\Scal_i\geq 0\ \forall\,i, \ \ \ \sum_iT^\Scal_i=W^\Scal,
\end{aligned}&&
\end{flalign}
or the dual problem
\begin{flalign}\label{sdp::dual}
\begin{aligned}
    \textbf{given}\ \ &\{p_i,C_i\} \\
    \textbf{minimize}\ \   &\lambda \\ 
    \textbf{subject to}\ \  &p_i\,C_i^{\otimes k}\leq \lambda\,\overline{W}^\Scal \ \ \forall\,i,
\end{aligned}&&
\end{flalign}
the latter of which can be straightforwardly phrased as an SDP by absorbing the coefficient $\lambda$ into $\overline{W}^\Scal$, which is the dual affine of process matrix $W^\Scal$.~\cite{bavaresco20} For parallel processes $W^\text{PAR}$, their dual affine $\overline{W}^\text{PAR}$ is a general $k$-partite quantum channel, for sequential processes $W^\text{SEQ}$, their dual affine $\overline{W}^\text{SEQ}$ is a $k$-partite channel with memory, and for general processes $W^\text{GEN}$, their dual affine $\overline{W}^\text{GEN}$ is a $k$-partite non-signaling channel. We refer to Ref.~\cite{bavaresco20} for more details.
In principle, any solution of the primal problem gives an upper bound to $P^\Scal$ and any solution of the dual problem, a lower bound. From strong duality, it is guarantee that the optimal solution of both problems coincides. Both problems are used in the computer-assisted-proof method applied in this paper.

%%%%%%%%%%%%%%%%%%%%%%%%%%%%%%%%%%%%%%%%%%%%%%%%%%%%%%%%%%%%%%%%%%%%%%%%%%%%%%%%%%%%%%%%%%%%%%%%%%%%%%%%%%%%
%%%%%%%%%%%%%%%%%%%%%%%%%%%%%%%%%%%%%%%%%%%%%%%%%%%%%%%%%%%%%%%%%%%%%%%%%%%%%%%%%%%%%%%%%%%%%%%%%%%%%%%%%%%%

\section{Proof of Theorem~\ref{thm::unitarygroup}}\label{app::proofunitarygroup}
\setcounter{theorem}{0}

We start this section with Lemma~\ref{lemma:bruto}, which plays a main role in the proof of Thm~\ref{thm::unitarygroup}, and may be of independent interest. The theorems presented in this section employ methods that are similar to the ones in Ref.~\cite{dariano01,bisio13}, which exploit the covariance of processes to parallelize strategies.
    
\begin{lemma}\label{lemma:bruto}
    Let $\{T_U\}_U$, $T_U\in\L(\H^I\otimes\H^O)$, be a general $k$-slot tester associated with the general process $W\coloneqq\sum_U T_U$, that respects the commutation relation 
    \begin{equation}
        W^{IO} \, \left(\id \otimes U^{\otimes k}\right)^{IO} = 
        \left(\id \otimes U^{\otimes k}\right)^{IO}\, W^{IO} ,
    \end{equation}
    for every unitary operator $U\in\L(\mathbb{C}^d)$ from a set $\{U\}_U$.

    Then, there exists a parallel $k$-slot tester $\{T_U^\text{PAR}\}_i$ such that
    \begin{equation}
        \tr\Big(T^\text{PAR}_U \dketbra{U}{U}^{\otimes k} \Big) =
        \tr\Big( T_U \dketbra{U}{U}^{\otimes k}\Big)
        \quad \forall \ U\in\{U\}_U.
    \end{equation}

    Moreover, this parallel tester can be written as $T_U^\text{PAR}=\rho^{I'I}*M_U^{I'O}$ where
    $\H^{I'}$ is an auxiliary space which is isomorphic to $\H^I$, $\rho\in\L(\H^I\otimes\H^{I'})$ is a quantum state defined by
    \begin{equation}
        \rho^{I'I} \coloneqq {\sqrt{W}^T}^{I'I}\,\dketbra{\id}{\id}^{I'I}\,{\sqrt{W}^T}^{I'I} ,
    \end{equation}
    and $\{M_U\}_U$ is a POVM defined by%
\interfootnotelinepenalty=10000
\footnote{Here, $\sqrt{W}^{-1}$ stands for the inverse of $\sqrt{W}$ on its range. If the operator $W$ is not full-rank, the composition $W\,W^{-1}=:\Pi_W$ is not the identity $\id$ but the projector onto the subspace spanned by the range of $\sqrt{W}$. Due to this technicality, when the operator $W$ is not full-rank, we should define the measurements as
$ M_U^{I'O} \coloneqq {\sqrt{W}^{-1}}^{I'O}\, {T_U}^{I'O}\, {\sqrt{W}^{-1}}^{I'O} + \frac{1}{N}(\id-W W^{-1})$. With that, the proof written here also applies to the case where the operator $W$ is not full-rank.}
    \begin{equation}
        M_U^{I'O} \coloneqq {\sqrt{W}^{-1}}^{I'O}\, {T_U}^{I'O}\, {\sqrt{W}^{-1}}^{I'O}.
    \end{equation}
\end{lemma} 

\begin{proof}
We start our proof by verifying that $\rho\in\L(\H^{I'}\otimes\H^{I})$ is a valid quantum state. The operator $\rho$ is positive semidefinite because it is a composition of positive semidefinite operators and the normalization condition follows from
\begin{align}
    \tr(\rho) =&\tr\left({\sqrt{W}^T}^{I'I}\,\dketbra{\id}{\id}^{I'I}\,{\sqrt{W}^T}^{I'I} \right)   \\
    =&\tr\left({W^T}^{I'I}\dketbra{\id}{\id}^{I'I}\right)  \\
    =&\tr\left({W}^{I'I}{\dketbra{\id}{\id}^T}^{I'I}\right)  \\
    =&\tr\left({W}^{I'I}\dketbra{\id}{\id}^{I'I}\right)  \\
    =& 1,
\end{align}
where the last equation holds because, since $W$ is a general process, it satisfies $\tr(WC)=1$ for any $C$ which is the Choi operator of a channel.

Let us now verify that the set of operators $\{ M_U\}_U$ forms a valid POVM. For that it is enough to recognize that all operators $M_U$ are compositions of positive semidefinite operators that add up to the identity, according to%

\begin{align}
    \sum_U M_U
    =&{\sqrt{W}^{-1}}^{I'O}\, {\sum_U T_U}^{I'O}\, {\sqrt{W}^{-1}}^{I'O}\\
    =&{\sqrt{W}^{-1}}^{I'O}\, W^{I'O}\, {\sqrt{W}^{-1}}^{I'O}\\
    =&\id^{I'O}.
\end{align}
The relation $\sqrt{W}^{-1}\,W\,\sqrt{W}^{-1}=\id$ can be shown by writing $W$ in an orthonormal basis as $W=\sum_i \alpha_i \ketbra{i}{i}$ and $\sqrt{W}^{-1}=\sum_i\alpha_i^{-1/2}\ketbra{i}{i}$.

Recall that for any unitary operator $U$, we have the identity $\dketbra{U}{U}^T=\dketbra{U^*}{U^*}$, and if $C^{IO}$ is the Choi operator of a linear map $\map{C}:\L(\H^I)\to\L(\H^O)$, $\rho^{I'I}\in\L(\H^{I'}\otimes\H^I)$, it holds that $\rho^{I'I}*C^{IO}= \left(\map{\id}\otimes\map{C}(\rho^{I'I})\right)^{I'O}$. In addition, if a diagonalizable operator $W^{IO}$ commutes with $\left(\id \otimes U^{\otimes k}\right)^{IO}$, its positive semidefinite square root $\sqrt{W}$ also commutes with%
\interfootnotelinepenalty=10000
\footnote{Indeed, two diagonalizable operators $A$ and $B$ commute if and only if they are diagonal in the same basis. Now, since $A=\sum_i\alpha_i\ketbra{i}{i}$, its square root, $\sqrt{A}\coloneqq\sum_i\sqrt{\alpha_i}\ketbra{i}{i}$, is by definition also diagonal in the same basis. Hence, if $A$ commutes with $B$, $\sqrt{A}$ also commutes with $B$.}%
$\left(\id \otimes U^{\otimes k}\right)^{IO}$; hence, we have
\begin{equation} \label{eq:root_commute}
    \sqrt{W}\,(\id\otimes {U}^{\otimes k})=(\id\otimes {U}^{\otimes k})\, {\sqrt{W}}.
\end{equation}
By taking the complex conjugation on both sides of Eq.\,\eqref{eq:root_commute} and exploiting the fact that $\sqrt{W}=\sqrt{W}^\dagger$ implies $\sqrt{W}^T=\sqrt{W}^*$, it holds that
\begin{equation}
    {\sqrt{W}^T}\,(\id\otimes {U^*}^{\otimes k})=(\id\otimes {U^*}^{\otimes k})\,{\sqrt{W}^T} .
\end{equation}
With these identities in hand, we can evaluate the link product $\rho^{I'I}*\left(\dketbra{U^{\otimes k}}{U^{\otimes k}}^T\right)^{IO}$, which will be used in the next step of the proof, to obtain
\begin{align}
    \rho^{I'I} * \left({\dketbra{U}{U}^{\otimes k}}^T\right)^{IO}
    =& 
    \rho^{I'I} * \left({\dketbra{U^*}{U^*}^{\otimes k}}\right)^{\, IO} \\
    =& 
    \left[(\id\otimes {U^*}^{\otimes k}) 
    \,\rho\, 
    (\id\otimes {U^T}^{\otimes k})\right]^{I'O} \\
    =& 
    \left[(\id\otimes {U^*}^{\otimes k}) 
    \,\sqrt{W}^T\, 
    \dketbra{\id}{\id} 
    \,{\sqrt{W}^T}\, 
    (\id\otimes {U^T}^{\otimes k})\right]^{I'O} \\
    =& 
    \left[{\sqrt{W}^T}\, 
    (\id\otimes {U^*}^{\otimes k}) 
    \,\dketbra{\id}{\id}\, 
    (\id\otimes {U^T}^{\otimes k})\, 
    {\sqrt{W}^T}\right]^{I'O} \\
    =& 
    \left({\sqrt{W}^T}
    \,\dketbra{U^*}{U^*}^{\otimes k}\,
    {\sqrt{W}^T} \right)^{I'O}. \label{eq::trick}
\end{align}

We now finish the proof by verifying that
\begin{align}
    \tr\left(T_U^\text{PAR} \, \dketbra{U}{U}^{\otimes k}\right)  
    &= 
    \tr\left[ 
    (\rho^{I'I}*M_U^{I'O}) \, 
    {\dketbra{U}{U}^{\otimes k}}^{IO}
    \right] \\
    &= 
    (\rho^{I'I}*M_U^{I'O}) * 
    {(\dketbra{U}{U}^{\otimes k})^T}^{\,IO} \\
    &= M_U^{I'O} * 
    \left(\rho^{I'I} * {\dketbra{U^*}{U^*}^{\otimes k}}^{\,IO}\right) \\
    &= 
    M_U^{I'O} * 
    \left({\sqrt{W}^T} \,\dketbra{U^*}{U^*}^{\otimes k}\, {\sqrt{W}^T} \right)^{I'O}
    \quad\quad\quad \text{(applying Eq.~\eqref{eq::trick})} \\
    &= 
    \tr\left[M_U^{I'O} {\left(
    {\sqrt{W}^T} 
    \,\dketbra{U^*}{U^*}^{\otimes k}\, 
    {\sqrt{W}^T} \right)^{T}}^{I'O}\right]\\
    &= 
    \tr\left[M_U^{I'O} \left(
    {\sqrt{W}} 
    \,\dketbra{U}{U}^{\otimes k}\, 
    {\sqrt{W}} \right)^{I'O}\right]\\
    &= 
    \tr\left[ \left(
    {\sqrt{W}^{-1}}\,T_U\,{\sqrt{W}^{-1}}\right)^{I'O} 
    \left({\sqrt{W}}\,\dketbra{U}{U}^{\otimes k}\,{\sqrt{W}} \right)^{I'O}\right]\\
    &= \tr\left(T_U \, \dketbra{U}{U}^{\otimes k}\right).
\end{align}	
\end{proof}

Now, we prove Theorem~\ref{thm::unitarygroup}.

\begin{theorem}
    Let $\mathcal{E}=\{p_U,\dketbra{U}{U}\}_U$ be an ensemble of unitary channels where the set of unitary operators $U\in\L(\mathbb{C}^d)$, $\{U\}_U$ forms a group up to a global phase---that is, there exist real numbers $\phi_i$ such that
    \begin{itemize}
	    \item  $e^{i \phi_\id}\id \in \{U\}_U$
	    \item If $A \in \{U\}_U$, then  $e^{i\phi_A}A^{-1}\in \{U\}_U$
	    \item If $A,B \in \{U\}_U$, then  $e^{i\phi_{AB}}AB\in \{U\}_U$,
    \end{itemize} 
    and the distribution $\{p_U\}_U$ is uniform---that is, if the set has $N$ elements, $p_U=\frac{1}{N}$.%, and if the set has a countable infinite number of elements, $\{p_U\}_U$ corresponds to the Haar measure.

    Then, for any number of uses $k$ and every general tester $\{T^\text{GEN}_U\}_U$, $T^\text{GEN}_U\in\mathcal{L}(\H^I\otimes\H^O)$ there exists a parallel tester $\{T^\text{PAR}_U\}_U$, $T^\text{PAR}_U\in\mathcal{L}(\H^I\otimes\H^O)$, such that
    \begin{equation}
	    \frac{1}{N} \sum_{U\in\{U\}_U} \tr\left(T^\text{GEN}_U \dketbra{U}{U}^{\otimes k}\right) 
	    =
	    \frac{1}{N} \sum_{U\in\{U\}_U} \tr\left(T^\text{PAR}_U \dketbra{U}{U}^{\otimes k}\right).
    \end{equation}
\end{theorem}    

Before presenting the proof, we recall that unitary operators that are equivalent up to a global phase represent equivalent unitary channels. That is, if $U'=e^{i \phi} U:\H^I\to\H^O$ is a linear operator, its associated map is given by 
\begin{align}
    \map{U'}(\rho)&=U'\rho {U'}^\dagger \\
    &=e^{i\phi}e^{-i\phi}U\rho {U}^\dagger \\
    &=U\rho {U}^\dagger \\
    &=\map{U}(\rho)
\end{align}
and its Choi operator $\dketbra{U}{U}$ respects
\begin{align}
    \dketbra{U}{U}=\dketbra{U'}{U'}.
\end{align}
Due to this fact, the two sets of operators $\{U_i\}_i$ and $\{e^{i\phi_i}U_i\}_i$ represent the same set of quantum channels.

\begin{proof}
The proof goes as follows: we start by using the general tester $\{T_U^\text{GEN}\}_U$ to construct another general tester $\{T_U\}_U$ which obeys
    \begin{equation}
	    \frac{1}{N} \sum_{U\in\{U\}_U} \tr\left(T^\text{GEN}_U \dketbra{U}{U}^{\otimes k}\right) 
	    =
	    \frac{1}{N} \sum_{U\in\{U\}_U} \tr\left(T_U \dketbra{U}{U}^{\otimes k}\right).
    \end{equation}
Then, we prove that the general tester $\{T_U\}_U$ we defined respects the hypothesis of Lemma~\ref{lemma:bruto}. Hence, there exists a parallel tester $\{T_U^\text{PAR}\}_U$ which is equivalent to $\{T_U\}_U$ when acting on the set of unitary operators $\{U\}_U$.

Let us start by defining the general tester $\{T_U\}_U$ as:
\begin{equation}
	T_U \coloneqq 
	\frac{1}{N} \sum_{V\in\{U\}_U} 
	\left( \id^{I}\otimes V^{\dagger^{\otimes k}}\right)^{IO}
	\,{T^\text{GEN}_{VU}}^{\,IO}\, 
	\left( \id^{I}\otimes V^{{\otimes k}}\right)^{IO},
\end{equation}
where $VU$ stands for the standard operator composition up to a global phase, that is, if $VU$ is not in the set $\{U\}_U$, we pick $e^{i\phi_{VU}} VU$, which is ensured to be an element of $\{U\}_U$.
Before proceeding, we should verify that the set of operators $\{T_U\}_U$ is indeed a valid general tester. Note that since $T_U$ is a composition of positive semidefinite operators, it holds that $T_U\geq0$ for every $U$. We now show that $W\coloneqq\sum_U T_U$ is a valid general process. First, note that
\begin{align}
	W &\coloneqq \sum_U T_U \\ 
	&= 
	\sum_U \frac{1}{N} \sum_V 
	(\id\otimes V^{\dagger^{\otimes k}}) 
	\,T^\text{GEN}_{VU}\, 
	(\id\otimes V^{{\otimes k}}) \\
    &= 
    \frac{1}{N} \sum_U \sum_V 
    (\id\otimes V^{\dagger^{\otimes k}}) 
    \,T^\text{GEN}_{V(V^{-1}U)}\, 
    (\id\otimes V^{{\otimes k}}) \label{eq:group_trick}  \\
    &= 
    \frac{1}{N} \sum_U \sum_V 
    (\id\otimes V^{\dagger^{\otimes k}}) 
    \,T^\text{GEN}_{U}\, 
    (\id\otimes V^{{\otimes k}}) \\
    &= 
    \frac{1}{N} \sum_V 
    (\id\otimes V^{\dagger^{\otimes k}}) 
    \,\sum_U T^\text{GEN}_{U}\, 
    (\id\otimes V^{{\otimes k}}) \\
    &= 
    \frac{1}{N} \sum_V 
    (\id\otimes V^{\dagger^{\otimes k}}) 
    \,W^\text{GEN}\, 
    (\id\otimes V^{{\otimes k}}),
\end{align}
where $W^\text{GEN}\coloneqq\sum_U T^\text{GEN}_U$ and, in Eq.~\eqref{eq:group_trick}, we have used the change of variable $U\mapsto V^{-1} U $, which does not affect the sum because the set $\{U\}_U$ is a group. 

Note also that, if $C^{IO}$ is the Choi operator of a quantum channel, the operator defined by
\begin{equation}
    C'^{IO}\coloneqq \frac{1}{N} \sum_V (\id\otimes V^{{\otimes k}})^{IO} \,C^{IO}\, (\id\otimes V^{\dagger^{\otimes k}})^{IO},
\end{equation}
is a valid channel since it is positive semidefinite and $\tr_O(C'^{IO})=\tr_O(C^{IO})=\id^I$. It then follows that, for every quantum channel of the form $C=\bigotimes_{i=1}^k C_i^{I_iO_i}$, we have
\begin{align}
    \tr(W^{IO}C^{IO}) 
    &= 
    \frac{1}{N} \tr\left[\sum_V 
    (\id\otimes V^{\dagger^{\otimes k}}) 
    \,W^\text{GEN}\, 
    (\id\otimes V^{{\otimes k}}) \,C\right] \\
    &= 
    \frac{1}{N} \tr\left[W^\text{GEN} \sum_V 
    (\id\otimes V^{{\otimes k}}) 
    \,C\, 
    (\id\otimes V^{\dagger^{\otimes k}})\right] \\
    &=
    \tr(W^{IO}C'^{IO}) \\
    &=1, 
\end{align}
ensuring that $\{T_U\}_U$ is a valid general tester.

The next step is to show that the tester $\{T_U\}_U$ attains the same success probability for discriminating the ensemble $\mathcal{E}=\{p_U,\dketbra{U}{U}\}_U$ as the tester $\{T^\text{GEN}_U\}_U$. This claim follows from direct calculation, that is, 
\begin{align}
    \frac{1}{N} \sum_U \tr(T_U \dketbra{U}{U}^{\otimes k}) 
    &=
    \frac{1}{N^2} \sum_U\sum_V \tr\left[
	(\id\otimes V^{\dagger^{\otimes k}})\,
	T^\text{GEN}_{VU}\, 
	(\id\otimes V^{{\otimes k}})\,
	\dketbra{U}{U}^{\otimes k} \right]  \\
    &=
    \frac{1}{N^2} \sum_U\sum_V \tr\left[
	T^\text{GEN}_{VU}\, 
	(\id\otimes V^{{\otimes k}})\,
	\dketbra{U}{U}^{\otimes k}\,
	(\id\otimes V^{\dagger^{\otimes k}}) \right] \\
    &=
    \frac{1}{N^2} \sum_U\sum_V\tr\Big(
    T^\text{GEN}_{VU}\, 
    \dketbra{VU}{VU}^{\otimes k} \Big) \label{eq:group_trick2} \\
    &=
    \frac{1}{N^2} \sum_U\sum_V\tr\left(
	T^\text{GEN}_{V\left(V^{-1}U\right)}\,
	\dketbra{V\left(V^{-1}U\right)}{V\left(V^{-1}U\right)}^{\otimes k}\right) \\
    &=
    \frac{1}{N^2} \sum_U \sum_V \tr\Big(
	T^\text{GEN}_{U}\, 
	\dketbra{U}{U}^{\otimes k}\Big) \\
    &=
    \frac{1}{N} \sum_U \tr\Big(
	T^\text{GEN}_{U}\,
	\dketbra{U}{U}^{\otimes k}\Big). \label{eq::tgen=tu}
\end{align}

The final step is to verify that the process $W\coloneqq\sum_U T_U$ commutes with $\id\otimes U^{{\otimes k}}$ for every unitary operator $U\in\{U\}_U$ to ensure that the tester $\{T_U\}_U$ fulfills the hypothesis of Lemma~\ref{lemma:bruto}. Direct calculation shows that
\begin{align}
    (\id\otimes U^{{\otimes k}}) \,W\, (\id\otimes {U^\dagger}^{{\otimes k}})
    =&
    (\id\otimes U^{{\otimes k}}) 
    \frac{1}{N} \sum_V 
    (\id\otimes V^{\dagger^{\otimes k}})
    \,W^\text{GEN}\, 
    (\id\otimes V^{{\otimes k}})
    (\id\otimes {U^\dagger}^{{\otimes k}}) \\
    =&
    \frac{1}{N} \sum_V 
    \left(\id\otimes \left(UV^\dagger\right)^{\otimes k}\right)
    \,W^\text{GEN}\,
    \left( \id\otimes \left(VU^\dagger\right)^{{\otimes k}}\right) \\
    =&
    \frac{1}{N} \sum_V 
    \Big[\id\otimes (U(VU)^\dagger)^{\otimes k}\Big]
    \,W^\text{GEN}\, 
    \Big[\id\otimes ((VU)U^\dagger)^{{\otimes k}}\Big] \\
    =&
    \frac{1}{N} \sum_V 
    (\id\otimes {V^\dagger}^{\otimes k})
    \,W^\text{GEN}\,
    (\id\otimes V^{{\otimes k}}) \\
    =& W.
\end{align}
Hence, we have that
\begin{equation}
    W^{IO} \, \left(\id \otimes U^{\otimes k}\right)^{IO} = \left(\id \otimes U^{\otimes k}\right)^{IO}\, W^{IO},
\end{equation}
and by Lemma~\ref{lemma:bruto}, one can construct a parallel tester $\{T_U^\text{PAR}\}_U$ which respects
\begin{equation}
    \tr\Big(T^\text{PAR}_U \dketbra{U}{U}^{\otimes k} \Big) =
    \tr\Big( T_U \dketbra{U}{U}^{\otimes k}\Big)
    \quad \forall \ U\in\{U\}_U,
\end{equation}
and therefore, by applying Eq.~\eqref{eq::tgen=tu}, we have
\begin{equation}
    \frac{1}{N}\sum_U\tr\Big(T^\text{PAR}_U \dketbra{U}{U}^{\otimes k} \Big) =
    \frac{1}{N}\sum_U\tr\Big(T^\text{GEN}_U \dketbra{U}{U}^{\otimes k}\Big),
\end{equation}
concluding our proof.

\end{proof}

%%%%%%%%%%%%%%%%%%%%%%%%%%%%%%%%%%%%%%%%%%%%%%%%%%%%%%%%%%%%%%%%%%%%%%%%%%%%%%%%%%%%%%%%%%%%%%%%%%%%%%%%%%%%
%%%%%%%%%%%%%%%%%%%%%%%%%%%%%%%%%%%%%%%%%%%%%%%%%%%%%%%%%%%%%%%%%%%%%%%%%%%%%%%%%%%%%%%%%%%%%%%%%%%%%%%%%%%%

\section{Proof of Examples~\ref{ex::k2N4} and~\ref{ex::k2N8}}\label{app::proofseq_vs_par}
\setcounter{example}{0}

The examples in this section show the advantage of sequential strategies over parallel strategies in channel discrimination tasks that involve only unitary channels and using $k=2$ copies. In the examples of this section, general strategies cannot outperform sequential ones. We recall that in the following $\sigma_x$, $\sigma_y$, and $\sigma_z$ denote the Pauli operators and $H\coloneqq\ketbra{+}{0}+\ketbra{-}{1}$, where $\ket{\pm}\coloneqq \frac{1}{\sqrt{2}}(\ket{0}\pm \ket{1})$ denotes the Hadamard gate. In addition, if $A\in\mathcal{L}(\mathbb{C}^d)$ is a linear operator with spectral decomposition $A=\sum_i \alpha_i \ketbra{\psi_i}{\psi_i}$, its square root is defined as $\sqrt{A}:=\sum_i \sqrt{\alpha_i}\ketbra{\psi_i}{\psi_i}$.

We start by proving Example~\ref{ex::k2N4} from the main text. It concerns the discrimination of an ensemble composed of a uniform probability distribution and a set of unitaries that does not form a group.
\begin{example}%\label{ex::k2N4}
    The ensemble composed of a uniform probability distribution and $N=4$ qubit unitary channels given by $\{U_i\} = \{\id,\sqrt{\sigma_x},\sqrt{\sigma_y},\sqrt{\sigma_z}\}$, in a discrimination task that allows for $k=2$ copies, can be discriminated under a sequential strategy with success probability $P^\text{SEQ}=1$, while any parallel strategy yields $P^\text{PAR}<1$.
\end{example}

\begin{proof}
The sequential strategy that attains perfect discrimination is easily understood by realizing that when the $k=2$ copies of the unitaries $\{U_i\}$ are applied in sequence, one recovers  $U_i\,U_i = \sqrt{\sigma_i}\,\sqrt{\sigma_i} = \sigma_i$, where $\sigma_1=\id,\sigma_2=\sigma_x,\sigma_3=\sigma_y,\sigma_4=\sigma_z$. Therefore the task reduces to the discrimination of the four Pauli operators with $k=1$ copy, which can be perfectly realized with a two-qubit maximally entangled state and a Bell measurement.

In order to show that the probability $P^\text{PAR}$ of discriminating these unitary channels with $k=2$ copies in a parallel strategy is strictly less than one, we make use of the dual problem associated with the SDP that computes the maximal probability of successful discrimination, given in Eq.~\eqref{sdp::dual}. Hence, in order to obtain an upper bound for the maximal success probability $P^\text{PAR}$, it is enough to find a value $\lambda<1$ and the Choi state of a quantum channel $\overline{W}$, that is, $\overline{W}\geq0$ and $\tr_O(\overline{W})=\id^I$, that respect 
\begin{equation} \label{eq:upper_bound}
	\frac{1}{4}\dketbra{U_i}{U_i}^{\otimes 2}\leq \lambda \,\overline{W} \ \ \ \text{for } \ i\in\{1,2,3,4\}.
\end{equation}

Using the computer-assisted-proof method presented in Ref.~\cite{bavaresco20}, we obtain an operator $\overline{W}$ that satisfies all the quantum channel conditions exactly and, for $\lambda=\frac{9571}{1000}$, satisfies the inequality~\eqref{eq:upper_bound}. Hence,
\begin{equation}
P^\text{PAR}\leq\frac{9571}{10000}. 
\end{equation}
In our online repository~\cite{githubMTQ2} we present a Mathematica notebook that can be used to verify that $\overline{W}$ is a valid Choi state of a quantum channel.
\end{proof}

Another similar example with an interesting property is given by the unitary channel ensemble composed of a uniform probability distribution and $\{U_i\} = \{\id,{\sigma_x},{\sigma_y},\sqrt{\sigma_z}\}$.

To prove this example as well, let us start by constructing a perfect sequential strategy.
We start by noting that the four Bell states can be written as:
\begin{align*}
	&\ket{\phi^+}\coloneqq \frac{\ket{00}+\ket{11}}{\sqrt{2}}=\left(\id\otimes\id\right) \ket{\phi^+}
	, \quad\quad\quad\quad \ket{\phi^-}\coloneqq \frac{\ket{00}-\ket{11}}{\sqrt{2}}=\left(\id\otimes\sigma_z\right) \ket{\phi^+},\\
	&\ket{\psi^+}\coloneqq \frac{\ket{01}+\ket{10}}{\sqrt{2}}=\left(\id\otimes\sigma_x\right) \ket{\phi^+}, \quad\quad\quad\, \ket{\psi^-}\coloneqq \frac{\ket{01}-\ket{10}}{\sqrt{2}}=(-i)\left(\id\otimes\sigma_y\right) \ket{\phi^+}.
\end{align*}
Additionally, since we take $\sqrt{\sigma_z}=\ketbra{0}{0}+i\ketbra{1}{1}$, we can check that the state
\begin{equation}
\left( \id\otimes \sqrt{\sigma_z} \right) \ket{\phi^+} = \frac{\ket{00}+i\ket{11}}{\sqrt{2}}
\end{equation}
is orthogonal to $\ket{\psi^+}$ and $\ket{\psi^-}$. We will now exploit these identities to construct a sequential strategy that attains $P^\text{SEQ}=1$ with $k=2$ uses. 

The strategy goes as follows: Define the auxiliary space $\H^{\text{aux}_1}$ to be isomorphic to $\H^{I_1}$ and prepare the initial state $\rho\in\L(\H^{I_1}\otimes\H^{\text{aux}_1})$ as
\begin{equation}
	\rho^{I_1\,\text{aux}_1}\coloneqq\ketbra{\phi^+}{\phi^+}^{I_1\text{aux}_1}.
\end{equation}
The state $\rho^{I_1\text{aux}_1}$ will then be subjected to the first copy of a unitary channel $U_i$, leading to the output state $\rho^{O_1\,\text{aux}_1}$
\begin{equation}
\rho^{O_1\,\text{aux}_1} = \left(U_i \otimes \id^{\text{aux}_1} \right)\rho^{I_1\,\text{aux}_1}\left(U_i \otimes \id^{\text{aux}_1} \right)^\dagger
=
\left(U_i \otimes \id^{\text{aux}_1} \right)\ketbra{\phi^+}{\phi^+}^{I_1\,\text{aux}_1}\left(U_i \otimes \id^{\text{aux}_1}\right)^\dagger.
\end{equation}	
Then, we apply the following transformation on the state $\rho^{O_1\,\text{aux}_1}$. We perform a projective measurement with POVM elements given by
\begin{align}
	M_{\psi^+}&\coloneqq \ketbra{\psi^+}{\psi^+}  \\
	M_{\psi^-}&\coloneqq \ketbra{\psi^-}{\psi^-}  \\
	M_\phi&\coloneqq \ketbra{\phi^+}{\phi^+} +\ketbra{\phi^-}{\phi^-}  
\end{align}
and, after the measurement, re-prepare the quantum system in the state 
$\rho^{I_2\,\text{aux}_2}=\map{M_i}(\rho)=\sqrt{M_i}\rho\sqrt{M_i}^\dagger$ with probability $\tr\big(\map{M_i}(\rho)\big)=\tr(\rho M_i)$. This transformation can be described by a L\"uders instrument. It can be checked that, if $U_i=\sigma_x$, one obtains the outcome associated with $M_{\psi^+}$ with probability one. Similarly, if $U_i=\sigma_y$, one obtains the outcome associated with $M_{\psi^-}$ with probability one. Hence, in these two cases, we have perfect channel discrimination. Now, if we obtain the outcome associated with $M_\phi$, the unitary $U_i$ can be either $\id$ or $\sqrt{\sigma_z}$.
	
After performing the projective measurement with elements $\{M_{\psi^+},M_{\psi^-},M_\phi\}$ and a L\"uders instrument, the state $\rho^{I_2\,\text{aux}_2}$ is subjected to a second copy of $U_i$. Direct calculation shows that if $U_i=\id$, then after the use of the second copy of unitary $U_i$, the state $\rho^{O_2\,\text{aux}_2}$ of the system is
\begin{equation}
\left(\id\otimes\id\right)^2 \ket{\phi^+}=\ket{\phi^+}.
\end{equation}

If $U_i=\sqrt{\sigma_z}$ after the second use of the unitary $U_i$ the state $\rho^{O_2\,\text{aux}_2}$ of the system is
\begin{equation}
\left(\id\otimes\sqrt{\sigma_z}\right)^2 \ket{\phi^+}= \ket{\phi^-}.
\end{equation}

Since $\ket{\phi^+}$ and $\ket{\phi^-}$ are orthogonal, they can be discriminated with probability one. Hence, the set of unitary operators $\{U_i\}_{i=1}^4$ can be perfectly discriminated in a sequential strategy with $k=2$ copies.

Using the tester formalism, this sequential strategy would be presented in terms of a sequential tester $T^\text{SEQ}=\{T^\text{SEQ}_i\}$, which can be implemented by an input quantum state $\rho$, a quantum encoder channel $\map{E}$, and a quantum measurement $\{N_i\}_i$. For completeness, we now present an explicit sequential tester that attains $P^\text{SEQ}=1$.
As in the strategy described earlier, we set the initial state as $\rho^{I_1\,\text{aux}_1}\coloneqq\ketbra{\phi^+}{\phi^+}^{I_1\,\text{aux}_1}$. Now, instead of using an instrument, we define a quantum encoder channel
$\map{E}: \L(\H^{O_1}\otimes\H^{\text{aux}_1})\to\L(\H^{I_2}\otimes\H^{\text{aux}_2}\otimes\H^{{\text{aux}'_2}})$ as
\begin{equation}
	\map{E}(\rho) = \left(\sqrt{M_\phi}\,\rho\,\sqrt{M_\phi}^\dagger \right)\otimes M_\phi^{{\text{aux}'_2}} \, +\,
	\left(\sqrt{M_{\psi^+}}\,\rho\,\sqrt{M_{\psi^+}}^\dagger \right)\otimes M_{\psi^+}^{{\text{aux}'_2}} \, +\, 
	\left(\sqrt{M_{\psi^-}}\,\rho\,\sqrt{M_{\psi^-}}^\dagger \right) \otimes M_{\psi^-}^{{\text{aux}'_2}},
\end{equation}
so that $\rho^{I_2\,\text{aux}_2\,{\text{aux}'_2}}=\map{E}(\rho^{O_1\,\text{aux}_1})$. We finish our sequential tester construction by presenting quantum measurement given by operators $N_i\in\L(\H^{O_2}\otimes\H^\text{aux}\otimes\H^{{\text{aux}'_2}})$,
\begin{align}
	N_1 &\coloneqq \ketbra{\phi^+}{\phi^+} \otimes M_\phi   \\
	N_2 &\coloneqq \ketbra{\phi^-}{\phi^-} \otimes M_\phi   \\
	N_3 &\coloneqq \id \otimes M_{\psi^+}   \\
	N_4 &\coloneqq \id \otimes M_{\psi^-}. 
\end{align}
In this way, if $E$ is the Choi operator of the channel $\map{E}$, the sequential tester with elements $T^\text{SEQ}_i\coloneqq\rho*E*N_i^T$ respects $\sum_i\tr(T^\text{SEQ}_i \dketbra{U_i}{U_i}^{\otimes 2})=1$.

In order to show that the probability $P^\text{PAR}$ of discriminating these unitary channels with $k=2$ copies in a parallel strategy is strictly less than one, we apply the method of computer-assisted proof again to obtain the upper bound of
\begin{equation}
P^\text{PAR}\leq\frac{9741}{10000}. 
\end{equation}

An interesting property of this example is that, with the above described sequential strategy, for the cases in which the unknown channel is either $U_2=\sigma_x$ or  $U_3=\sigma_y$, a conclusive answer is achieved after only one use of the unknown channel. From the uniform probability of the ensemble, we know that this scenario would occur with probability $\frac{1}{2}$. Only when the unknown channel is  either $U_1=\id$ or  $U_4=\sqrt{\sigma_y}$ is it necessary to use the second copy of the unknown channel to arrive at a conclusive answer. This scenario would also occur with probability $\frac{1}{2}$. If one considers this discrimination task as being performed repeatedly, always drawing the unknown channel with uniform probability from the ensemble $\{U_i\} = \{\id,{\sigma_x},{\sigma_y},\sqrt{\sigma_z}\}$, then one can see that, on average over the multiples runs of the task, perfect discrimination will be achieved using only 1.5 copies of the unknown channel under this sequential strategy. 
\\

We now prove Example~\ref{ex::k2N8} from the main text. It concerns the discrimination of an ensemble composed of a non-uniform probability distribution and a set of unitaries that forms a group.

\begin{example}
    Let $\{U_i\}=\{\id,\,\sigma_x,\,\sigma_y,\,\sigma_z,\,H,\,\sigma_xH,\,\sigma_yH,\,\sigma_zH\}$ be a tuple of $N=8$ unitary channels that forms a group up to a global phase, and let $\{p_i\}$ be a probability distribution in which each element $p_i$ is proportional to the $i$-th digit of the number $\pi\approx3.1415926$, that is, $\{p_i\}=\{\frac{3}{31},\frac{1}{31},\frac{4}{31},\frac{1}{31},\frac{5}{31},\frac{9}{31},\frac{2}{31}.\frac{6}{31}\}$. For the ensemble $\{p_i,U_i\}$, in a discrimination task that allows for $k=2$ copies, sequential strategies outperform parallel strategies, i.e., $P^\text{PAR}<P^\text{SEQ}$.
\end{example}

\begin{proof}
The first step of the proof is to ensure that the tuple $\{\id,\,\sigma_x,\,\sigma_y,\,\sigma_z,\,H,\,\sigma_xH,\,\sigma_yH,\,\sigma_zH\}$ forms a group up to a global phase. This is done by direct inspection. The second step of the proof is to ensure that there is a sequential strategy which outperforms any parallel one. We accomplish this step with the aid of the computer-assisted-proof methods presented in Ref.~\cite{bavaresco20}. These methods allow us to compute rigorous and explicit upper and lower bounds for the maximal probability of success under parallel and sequential strategies. We obtain
\begin{equation}
  \frac{8196}{10000} < P^\text{PAR} < \frac{8197}{10000}  < P^\text{SEQ} < \frac{8198}{10000}, 
\end{equation}
ensuring that $P^\text{PAR}<P^\text{SEQ}$.

The code used in the computer-assisted proof of the this example is publicly available at our online repository~\cite{githubMTQ2}, along with a Mathematica notebook, which shows that this set of unitaries forms a group.
\end{proof}

%%%%%%%%%%%%%%%%%%%%%%%%%%%%%%%%%%%%%%%%%%%%%%%%%%%%%%%%%%%%%%%%%%%%%%%%%%%%%%%%%%%%%%%%%%%%%%%%%%%%%%%%%%%%
%%%%%%%%%%%%%%%%%%%%%%%%%%%%%%%%%%%%%%%%%%%%%%%%%%%%%%%%%%%%%%%%%%%%%%%%%%%%%%%%%%%%%%%%%%%%%%%%%%%%%%%%%%%%

\section{Proof of Example~\ref{ex::k3N4}%and~\ref{ex::k3N8}
}\label{app::proofgen_vs_seq_vs_par}
\setcounter{example}{2}

The Example in this Section shows the advantage of general strategies over sequential strategies and of sequential strategies over parallel strategies in channel discrimination tasks that only involve unitary channels and using $k=3$ copies.

We start by proving Example~\ref{ex::k3N4} from the main text. It concerns the discrimination of an ensemble composed of a uniform probability distribution and a set of unitaries that does not form a group. We recall that, for the following, we define $H_y\coloneqq\ketbra{+_y}{0}+\ketbra{-_y}{1}$, where $\ket{\pm_y}\coloneqq \frac{1}{\sqrt{2}}(\ket{0}\pm i\ket{1})$, and $H_P\coloneqq\ketbra{+_P}{0}+\ketbra{-_P}{1}$, where $\ket{+_P}\coloneqq \frac{1}{5}(3\ket{0} + 4\ket{1})$ and $\ket{-_P}\coloneqq \frac{1}{5}(4\ket{0} - 3\ket{1})$.

\begin{example}%\label{ex::k3N4}
    For the ensemble composed of a uniform probability distribution and $N=4$ unitary channels given by $\{U_i\} = \{\sqrt{\sigma_x},\sqrt{\sigma_z},\sqrt{H_y},\sqrt{H_P}\}$, in a discrimination task that allows for $k=3$ copies, general strategies outperform sequential strategies, and sequential strategies outperform parallel strategies. Therefore, the maximal probabilities of success satisfy the strict hierarchy $P^\text{PAR}<P^\text{SEQ}<P^\text{GEN}$.
\end{example}

\begin{proof}
The proof follows from the direct application of the computer-assisted methods presented in Ref.~\cite{bavaresco20}. These methods allow us to find explicit and exact parallel/sequential/general testers that attain a given success probability, ensuring, then, a lower bound for the maximal success probability for its class. In addition, we can obtain an exact parallel/sequential/general upper bound, given the SDP dual formulation. The code used to obtain the computer-assisted proof of the present example is publicly available in our online repository~\cite{githubMTQ2}.

The computed bounds for the maximal probability of successful discrimination are
\begin{equation}
\begin{alignedat}{3}
  \frac{9570}{10000} < P^\text{PAR} < \frac{9571}{10000} & && && \\
  & < \frac{9876}{10000} < P^\text{SEQ} < \frac{9877}{10000} && && \\
  & && && < \frac{9881}{10000} < P^\text{GEN} < \frac{9882}{10000},
\end{alignedat}
\end{equation}
showing the advantage of strategies that apply indefinite causal order over ordered ones and proving a strict hierarchy between strategies for the discrimination of a set of unitary channels.
\end{proof}

%%%%%%%%%%%%%%%%%%%%%%%%%%%%%%%%%%%%%%%%%%%%%%%%%%%%%%%%%%%%%%%%%%%%%%%%%%%%%%%%%%%%%%%%%%%%%%%%%%%%%%%%%%%%
%%%%%%%%%%%%%%%%%%%%%%%%%%%%%%%%%%%%%%%%%%%%%%%%%%%%%%%%%%%%%%%%%%%%%%%%%%%%%%%%%%%%%%%%%%%%%%%%%%%%%%%%%%%%

\section{Proof of Theorem~\ref{thm::switchlike}}\label{app::proofswitchlike}
\setcounter{theorem}{3}

In this section, we prove Thm.~\ref{thm::switchlike} from the main text, which concerns the inability of switch-like strategies to outperform sequential strategies on channel discrimination tasks that involve only unitary channels. We also repeat here Fig.~\ref{fig:SL} from the main text for convenience of the reader.

\setcounter{figure}{1}

\begin{figure*}%[h!]
\begin{center}
	\includegraphics[width=\textwidth]{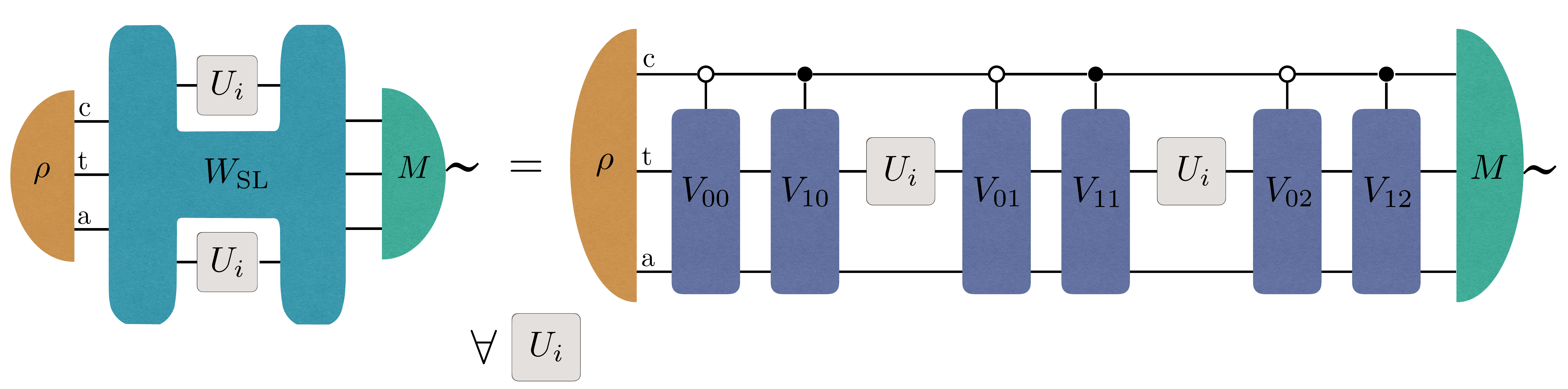}
	\caption{Illustration of a $2$-copy sequential strategy that attains the same probability of successful discrimination as any $2$-copy switch-like strategy, for all sets of unitary channels $\{U_i\}_{i=1}^N$. Line ``c'' represents a control system, ``t'' represents a target system, and ``a'' represents an auxiliary system. Both strategies can be straightforwardly extended to $k$ copies.}
\end{center}
\end{figure*}

\begin{theorem}%\label{thm::switchlike}
The action of the switch-like superchannel on $k$ copies of a unitary channel can be equivalently described by a sequential circuit that acts on $k$ copies of the same unitary channel.
    
Consequently, in a discrimination task involving the ensemble $\mathcal{E}=\{p_i,U_i\}_i$ composed of $N$ unitary channels and some probability distribution and that allows for $k$ copies, for every switch-like tester $\{T^\text{SL}_i\}_i$, there exists a sequential tester $\{T^\text{SEQ}_i\}_i$ that attains the same probability of success, according to
    \begin{equation}
	    \sum_{i=1}^N p_i \tr\Big(T^\text{SL}_i \dketbra{U_i}{U_i}^{\otimes k}\Big) 
    	=
	    \sum_{i=1}^N p_i \tr\Big(T^\text{SEQ}_i \dketbra{U_i}{U_i}^{\otimes k}\Big).
	\end{equation}
\end{theorem}

In order to provide a better intuition on this result,
before presenting the formal definition of the switch-like process with $k$ slots and proving Theorem~\ref{thm::switchlike} in full generality, we present a proof for the $k=2$ case that is illustrated in Fig.~\ref{fig:SL} in the main text. 
	
For the case $k=2$, the switch-like superchannel transforms a pair of unitary channels $\{U_1,U_2\}$ into one unitary channel, according to
\begin{align}
\begin{split}
    \Wcal_\text{SL}(U_1,U_2)\coloneqq &\ketbra{0}{0}^c \otimes V_{02}\,({U_2}\otimes\id^a)\,V_{01}\,({U_1}\otimes\id^a)\,V_{00} 
    + \ketbra{1}{1}^c\otimes V_{12}\,({U_1}\otimes\id^a)\,V_{11}\,({U_2}\otimes\id^a)\,V_{10},
\end{split}
\end{align}
where $\id^a$ is the identity operator acting on the auxiliary system and $V_{\pi i}$ are fixed unitary operators. Note that, if $U_1=U_2=U$, we have
\begin{align}
\begin{split}
    \Wcal_\text{SL}(U,U)= &\ketbra{0}{0}^c \otimes V_{02}\,({U}\otimes\id^a)\,V_{01}\,({U}\otimes\id^a)\,V_{00} 
    + \ketbra{1}{1}^c\otimes V_{12}\,({U}\otimes\id^a)\,V_{11}\,({U}\otimes\id^a)\,V_{10}.
\end{split}
\end{align}
We now define a controlled version of the unitary operators $V_{0i}$ as
\begin{equation}
    V^\text{ctrl}_{0 i}\coloneqq\ketbra{0}{0}^c\otimes V_{0i} + \ketbra{1}{1}^c \otimes \id^t \otimes \id^a.
\end{equation}
and a controlled version of $V_{1 i}$ as
\begin{equation}
    V^\text{ctrl}_{1 i}\coloneqq\ketbra{0}{0}^c\otimes \id^t \otimes \id^a + \ketbra{1}{1}^c \otimes V_{1i}.
\end{equation}

We first note that due to the orthogonality of $\ket{0}$ and $\ket{1}$, we have  $V^\text{ctrl}_{1i}V^\text{ctrl}_{0i}= \ketbra{0}{0}^c\otimes V_{0i} + \ketbra{1}{1}^c \otimes V_{1i}$. Hence, a direct calculation shows that 
\begin{align}
    V^\text{ctrl}_{12}V^\text{ctrl}_{02}\,(\id^c\otimes{U}\otimes\id^a)\ \cdot\ V^\text{ctrl}_{11}V^\text{ctrl}_{01}\,(\id^c\otimes{U}\otimes\id^a) \ \cdot\ V^\text{ctrl}_{10}V^\text{ctrl}_{00} =&
    \left(\ketbra{0}{0}^c\otimes V_{02} + \ketbra{1}{1}^c \otimes V_{12}\right)(\id^c\otimes{U}\otimes\id^a)\ \cdot \nonumber\\
    &\left(\ketbra{0}{0}^c\otimes V_{01} + \ketbra{1}{1}^c \otimes V_{11}\right)(\id^c\otimes{U}\otimes\id^a)\ \cdot \nonumber \\
    &\left(\ketbra{0}{0}^c\otimes V_{00} + \ketbra{1}{1}^c \otimes V_{10}\right)(\id^c\otimes{U}\otimes\id^a)\\
    =&\Wcal_\text{SL}(U,U).
\end{align}
This shows that, when $U_1=U_2=U$, a $2$-slot sequential circuit which performs the operations $V^\text{ctrl}_{12}V^\text{ctrl}_{02}$, $V^\text{ctrl}_{11}V^\text{ctrl}_{01}$, and $V^\text{ctrl}_{10}V^\text{ctrl}_{00}$ can perfectly simulate the two-slot switch-like superchannel. See Fig.~\ref{fig:SL} in the main text for an illustration.

Now, we extend this result to an arbitrary finite number of copies $k$.

\begin{definition}[Switch-like superchannel]
\label{def:SL}
Let $\{\pi\}_\pi$, $\pi\in\{0,\ldots,k!-1\}$ be a set in which each integer $\pi$ represents a permutation of the set $\{1,\ldots,k\}$ and $\sigma_\pi:\{1,\ldots,k\}\to\{1,\ldots,k\}$ be the permutation function such that, after permutation $\pi$, the element $i\in\{1,\ldots,k\}$ is mapped to $\sigma_\pi(i)$. 
The $k$-slot switch-like superchannel acts on a set of $k$ unitary operators $\{U_i\}_{i=1}^k$, $U_i:\H^{I_i}\to\H^{O_i}$ according to
\begin{equation}
    \mathcal{W}_\text{SL}(U_1,\ldots,U_k)\coloneqq \sum_{\pi=0}^{k!-1} \ketbra{\pi}{\pi}^c \otimes
    \left[V_{\pi k} \left(U_{\sigma_\pi(k)}\otimes\id^{a}\right)	V_{\pi(k-1)}
    \left(U_{\sigma_\pi(k-1)}\otimes\id^{a}\right)
    V_{\pi(k-2)}
    \ldots
    \left(U_{\sigma_\pi(1)}\otimes\id^{a}\right)	V_{\pi0}
    \right],
\end{equation}
where $\{V_{\pi n}\}_{\pi n}$ %_{\pi=0, i=1}^{k!-1,k}$
is a set of unitary operators defined as 
\begin{align}
    & V_{\pi0}:\H^{P_t}\otimes\H^{a}\to\H^{I_{\sigma_\pi(1)}}\otimes \H^{a} \\
    & V_{\pi n}:	\H^{I_{\sigma_\pi(i)}} \otimes\H^{a}\to\H^{O_{\sigma_\pi(i+1)}} \otimes \H^{a} \quad \quad \text{for}\, n\in\{1,\ldots,k-1\}\\
    & V_{\pi k}:\H^{I_{\sigma_\pi(k)}}\otimes\H^{a}\to\H^{F_t}\otimes \H^{a}
\end{align}	
\end{definition}

Here, we have defined the switch-like superchannel only by its action on unitary channels, without explicitly stating how the switch-like superchannel acts on general quantum operations%
\interfootnotelinepenalty=10000
\footnote{Interestingly, it can be proved that for the ($k=2$)-slot case, the action of the standard quantum switch superchannel (i.e. $V_{\pi n}=\id$) on unitary channels uniquely defines its action on general operations.~\cite{dong21} The possibility of extending this result for general switch-like superchannels is still open, but the existence of general switch-like superchannels is ensured by the construction presented via Eq.~\eqref{eq:SLexists}.}
or its process $W^\text{SL}\in\L\left(\H^P\otimes\H^I\otimes\H^O\otimes\H^F\right)$. In order to prove Theorem~\ref{thm::switchlike} and for the main purpose of this paper, knowing the action of switch-like superchannels only on unitary channels will be enough, but for the sake of concreteness, we also present an explicit process which implements the switch-like superchannel. For that, we define the process $W^\text{SL}:=\dketbra{U_\text{SL}}{U_\text{SL}}$ where 
\begin{equation} \label{eq:SLexists}
    U_\text{SL}:=\bigoplus_\pi V_{\pi k} V_{\pi k-1} \ldots V_{\pi 1} V_{\pi 0}.
\end{equation}
Following Lemma~1 in Ref.~\cite{yokojima20} (see also Theorem~2 of Ref.~\cite{araujo16}), one can verify that the process $W^\text{SL}$ acts on unitary operators according to the switch-like superchannel, as presented in Definition \ref{def:SL}.

\begin{proof}
We start our proof by defining the generalized controlled operation
\begin{equation}
    V^\text{ctrl}_n\coloneqq\sum_{\pi=0}^{k!-1} \ketbra{\pi}{\pi}^c \otimes V_{\pi n} \ \ \ \forall \ n\in\{0,\ldots,k\},
\end{equation}
which is a valid unitary operator since $V^\text{ctrl}_n\left(V^\text{ctrl}_n\right)^\dagger=\id$. Now, note that, due to orthogonality of the vectors $\ket{\pi}$, we have
\begin{equation}
    V^\text{ctrl}_k(\id^c\otimes U\otimes\id^a)V^\text{ctrl}_{(k-1)}(\id^c\otimes U\otimes\id^a)V^\text{ctrl}_{(k-2)}\ldots(\id^c\otimes U\otimes\id^a) V^\text{ctrl}_0 = \mathcal{W}_{SL}(U,\ldots,U).
\end{equation}
Hence, similarly to the $k=2$ case, a simple concatenation of the operators $V^\text{ctrl}_i$ provides a $k$-slot sequential quantum circuit that perfectly simulates the switch-like $k$-slot superchannel when all input unitary channels are equal. 

Since every sequential quantum circuit can be written as an ordered process $W^\text{SEQ}\in\L\left(\H^P\otimes\H^I\otimes\H^O\otimes\H^F\right)$~\cite{chiribella07}, when $k$ identical unitary operators $U$ are plugged into the process $W^\text{SEQ}$, the output operation is described by
\begin{equation}
    W^\text{SEQ}*\dketbra{U}{U}^{\otimes k} = W^\text{SL}*\dketbra{U}{U}^{\otimes k},
\end{equation}
where $*$ is the link product and $W^\text{SL}$ is a process associated with the switch-like superchannel. Hence, if
\begin{equation}
    T^\text{SL}_i\coloneqq \tr_{PF}[(\rho\otimes\id)W^\text{SL}(\id\otimes M_i)]
\end{equation}
is the tester associated to the switch-like strategy, then one can construct a sequential tester
\begin{equation}
    T^\text{SEQ}_i= \tr_{PF}[(\rho\otimes\id)W^\text{SEQ}(\id\otimes M_i)]
\end{equation} 
such that, for any unitary operator $U$, one has
\begin{equation}
 \tr\left(T^\text{SL}_i\, \dketbra{U}{U}^{\otimes k}\right)  = \tr\left(T^\text{SEQ}_i\, \dketbra{U}{U}^{\otimes k}\right),
\end{equation}
ensuring that there is always a sequential tester that performs as well as any switch-like one.

\end{proof}

%%%%%%%%%%%%%%%%%%%%%%%%%%%%%%%%%%%%%%%%%%%%%%%%%%%%%%%%%%%%%%%%%%%%%%%%%%%%%%%%%%%%%%%%%%%%%%%%%%%%%%%%%%%%
%%%%%%%%%%%%%%%%%%%%%%%%%%%%%%%%%%%%%%%%%%%%%%%%%%%%%%%%%%%%%%%%%%%%%%%%%%%%%%%%%%%%%%%%%%%%%%%%%%%%%%%%%%%%

\section{Upper bound}\label{app::upperbound}

We start this section by stating a lemma from Ref.~\cite{hashimoto10} which will be very useful for proving Thm.~\ref{thm::upper_bound}. We would also like to mention that step 2 of Ref.~\cite{korff04} and Thm.~3 of Ref.~\cite{chiribella06} are essentially equivalent to the lemma that we now state.
\begin{lemma}[Ref.~\cite{hashimoto10,korff04,chiribella06}]
\label{lemma:hashimoto}
Let $E\in\mathcal{L}(\mathbb{C}^d)$ be a positive semidefinite operator, $U_g\in\mathcal{L}(\mathbb{C}^d)$ be unitary operators, and $G=\{U_g\}_g$ be a (compact Lie) group of unitary operators up to a global phase.  It holds true that
\begin{equation}
    E\leq \gamma_G \int_\text{Haar} U_g\,E\, U_g^\dagger \text{d}g,
\end{equation}
with
\begin{equation}
    \gamma_G:=\sum_{\mu\in \text{irrep}\{G\}} d_\mu \min(m_\mu,d_\mu),
\end{equation}
where $\text{irrep}\{G\}$ is the set of all inequivalent irreducible representations (irreps) of $G$ in $\mathcal{L}(\mathbb{C}^d)$, $d_\mu$ is the dimension of the linear space corresponding to the irrep $\mu$ and $m_\mu$ is its associated multiplicity. 
\end{lemma}

We are now in conditions to prove Thm.~\ref{thm::upper_bound}.

\begin{theorem}[Upper bound for general strategies] %\label{thm::upper_bound}
Let $\mathcal{E}=\{p_i,U_i\}_{i=1}^N$ be an ensemble composed of $N$ $d$-dimensional unitary channels and a uniform probability distribution. The maximal probability of successful discrimination of a general strategy with $k$ copies is upper bounded by
\begin{equation}
    P^\text{GEN}\leq \frac{1}{N}\,\gamma(d,k),
\end{equation}
where $\gamma(d,k)$ is given by
\begin{equation}
   \gamma(d,k) \coloneqq 
{k+d^2-1\choose k}= \frac{(k+d^2-1)!}{k!(d^2-1)!}.
\end{equation}
\end{theorem}

\begin{proof}
The dual formulation of the channel discrimination problem [see Eq.~\eqref{sdp::dual}] guarantees that if $\overline{W}$ is the Choi operator of a non-signaling channel and the constraints
\begin{equation} \label{eq:upper_via_dual}
    p_i \dketbra{U_i}{U_i}^{\otimes k} \leq \lambda \overline{W} \quad \forall i\in\{1,\ldots,N\}
\end{equation}
are respected, the coefficient $\lambda$ is an upper bound for the maximal probability of successfully discriminating the ensemble $\mathcal{E}=\{p_i,U_i\}_{i=1}^N$ with $k$ copies.
Our proof consists in explicitly presenting the Choi operator of a non-signaling channel $\overline{W}$ and a real number $\lambda$ that respects the constraints of Eq.~\eqref{eq:upper_via_dual} for any ensemble $\mathcal{E}=\{p_i,U_i\}_{i=1}^N$ with $p_i=\frac{1}{N}$.

Consider the following ansatz: 
\begin{align}
    \overline{W}&:= \int_\text{Haar} \dketbra{U}{U}^{\otimes k} \text{d}U \\
    \lambda&:= \frac{1}{N}\sum_{\mu\in\text{irrep}\{SU(d)^{\otimes k}\}} d_\mu^2.
\end{align}
where $\text{irrep}\{SU(d)^{\otimes k}\}$ is the set of all inequivalent irreducible representations (irreps) of $SU(d)^{\otimes k}$ and $d_\mu$ is the dimension of the linear space corresponding to the irrep $\mu$.

Our first step is to show that $\overline{W}$ is the Choi operator of a non-signaling channel. The operator $\overline{W}$ is positive semidefinite because it is a convex mixture of positive semidefinite operators. Additionally, for any $d$-dimensional unitary operator $U$, we have that $\tr(\dketbra{U}{U})=d$. Hence, from the normalization of the Haar measure, we have that $\tr(\overline{W})=d^k$.

The last step to certify that $\overline{W}$ is indeed a non-signaling channel is then to guarantee that if $j\in\{1,\ldots,k\}$ stands for a slot of our process, we have that 
\begin{equation}
    _{O_j}\overline{W}=_{I_jO_j}\overline{W} \quad  \forall j\in\{1,\ldots,k\}
\end{equation}
where $\H_{I_j}$ and $\H_{O_j}$ correspond to the input and output space associated to the slot $j$ respectively.
Since for any $\dketbra{U}{U}\in\mathcal{L}\big(\H_{I_j}\otimes\H_{O_j}\big)$ we have
\begin{align}
\tr_{O_j} \Big( \dketbra{U}{U}\Big) 
&= \tr_{O_j} \Big( \big(  \id \otimes U \big)\,\dketbra{\id}{\id}\, \big( \id\otimes U^\dagger\big) \Big)\\
&=\tr_{O_j} \Big( \big(\id \otimes U^\dagger U \big)\, \dketbra{\id}{\id} \Big) \\
&= \tr_{O_j} \Big( \dketbra{\id}{\id} \Big) \\
&= \id_{I_j},
\end{align}
it holds that
\begin{align}
_{O_j}\overline{W} &=
\int_{\text{Haar}} \,_{O_j} \Big( \dketbra{U}{U}^{\otimes k}\Big) \text{d}U\\
&=\int_{\text{Haar}} \,_{O_i} 
\Big(\dketbra{U}{U}_{I_1O_1}\otimes\dketbra{U}{U}_{I_2O_2}\otimes \ldots \dketbra{U}{U}_{I_kO_k}\Big) \text{d}U \\
&=_{I_iO_i}\overline{W}.
\end{align}

The last step of the proof is then to certify that for $p_i=\frac{1}{N}$, we indeed have 
\begin{equation}
    p_i \dketbra{U_i}{U_i}^{\otimes k} \leq \lambda \overline{W}.
\end{equation}
for any set of unitary operators $\{U\}_{i=1}^N$.
First, we observe that due to the left and right invariance of the Haar measure, for any unitary operator $V\in SU(d)$, the operator $\overline{W}$ can be written as
\begin{align}
    \overline{W}:=&
    \int_\text{Haar} \dketbra{U}{U}^{\otimes k} \text{d}U\\
    =& 
    \int_\text{Haar} 
    \Big[\big(\id\otimes U\big)\dketbra{\id}{\id}\big( \id\otimes U^\dagger\big) \Big]^{\otimes k} \text{d}U \\
    =& 
    \int_\text{Haar} 
    \Big[\big(\id\otimes(UVU^\dagger) U\big)\dketbra{\id}{\id}\big( \id\otimes U^\dagger (UVU^\dagger)\big) \Big]^{\otimes k} \text{d}U \\
    =& 
    \int_\text{Haar} 
    \Big[\big(\id\otimes U\big)\dketbra{V}{V}\big( \id\otimes U^\dagger\big) \Big]^{\otimes k} \text{d}U.
\end{align}
Additionally, the set $\big\{\id^{\otimes k}\otimes U^{\otimes k}\big\}_{U\in SU(d)}$ is a compact Lie group. Moreover, the dimensions of the linear spaces associated with the irreducible representation of $\big\{\id^{\otimes k}\otimes U^{\otimes k}\big\}_{U\in SU(d)}$  coincide with the dimension of the irreducible representations of $\big\{U^{\otimes k}\big\}_{U\in SU(d)}$. Since $\min(d_\mu,m_\mu)\leq d_\mu$, Lemma~\ref{lemma:hashimoto} ensures that
\begin{align}
\frac{1}{N} \dketbra{U_i}{U_i}^{\otimes k}
&\leq \frac{1}{N} \left(\sum_{\mu\in\text{irrep}\{SU(d)^{\otimes k}\}} d_\mu^2\right) \int_\text{Haar} 
\Big[\big(\id\otimes U\big)\dketbra{U_i}{U_i}\big( \id\otimes U^\dagger\big) \Big]^{\otimes k} \text{d}U \\
&=\frac{1}{N} \left(\sum_{\mu\in\text{irrep}\{SU(d)^{\otimes k}\}} d_\mu^2\right) \int_\text{Haar} 
\Big[\big(\id\otimes U\big)\dketbra{\id}{\id}\big( \id\otimes U^\dagger\big) \Big]^{\otimes k} \text{d}U \\
&= \lambda \overline{W}.
\end{align}
Hence, $\lambda:=\frac{1}{N}\left(\sum_{\mu\in\text{irrep}\{SU(d)^{\otimes k}\}} d_\mu^2\right) $ is indeed an upper bound for the maximum probability of discriminating any set of $N$ $d$-dimensional unitary channels with $k$ copies with general strategies.

 We finish the proof by recognizing that, as proven in Ref.~\cite{schur} [p. 57, Eq.~(57)], we have the following identity:
\begin{equation}
    \sum_{\mu\in\text{irrep}\{SU(d)^{\otimes k}\}} d_\mu^2 ={k+d^2-1\choose k}.
\end{equation}
\end{proof}

\newpage

%%%%%%%%%%%%%%%%%%%%%%%%%%%%%%%%%%%%%%%%%%%%%%%%%%%%%%%%
%%%%%%%%%%%%%%%%%%%%%%%%%%%%%%%%%%%%%%%%%%%%%%%%%%%%%%%%

\twocolumngrid

%\bibliography{refs.bib}

%merlin.mbs aipnum4-1.bst 2010-07-25 4.21a (PWD, AO, DPC) hacked
%Control: key (0)
%Control: author (8) initials jnrlst
%Control: editor formatted (1) identically to author
%Control: production of article title (0) allowed
%Control: page (1) range
%Control: year (1) truncated
%Control: production of eprint (0) enabled
%

\end{document}